\documentclass[12pt]{article}%
\usepackage{amssymb}
\usepackage{amsfonts}
\usepackage{amsmath}
\usepackage{amsthm}
\usepackage{mathtools}
\usepackage{bm}
\usepackage[nohead]{geometry}
\usepackage[singlespacing]{setspace}
\IfFileExists{footmisc.sty}{\usepackage[bottom]{footmisc}}{}
\usepackage{indentfirst}
\usepackage{graphicx}%
\usepackage{rotating}
\usepackage{placeins}
\usepackage{enumerate}
\usepackage{booktabs}
\usepackage{subcaption}
\usepackage[flushleft]{threeparttable}
\usepackage{float}
\usepackage{verbatim}
\usepackage{epstopdf}
\usepackage{pdflscape}
\usepackage{multirow}
\usepackage{enumitem}
\usepackage{microtype}
\usepackage{bibunits}
\usepackage{siunitx}
\defaultbibliographystyle{plainnat}
\defaultbibliography{ref}

\newtheorem{theorem}{Theorem}[section]
\newtheorem{lemma}[theorem]{Lemma}
\newtheorem{example}[theorem]{Example}
\newtheorem{proposition}[theorem]{Proposition}
\newtheorem{corollary}[theorem]{Corollary}

\newtheorem{assumption}{Assumption}
\theoremstyle{remark}
\newtheorem{remark}[theorem]{Remark}

\newcommand{\E}{\mathbb{E}}
\newcommand{\Var}{\mathrm{Var}}
\newcommand{\Cov}{\mathrm{Cov}}

\newcommand{\1}{\mathbf{1}}

\newcommand{\R}{\mathbb{R}}
\newcommand{\T}{\mathcal{T}}

\DeclareMathOperator*{\argmin}{arg\,min}

\providecommand{\abs}[1]{\lvert#1\rvert}

\usepackage[usenames,dvipsnames]{color}
\usepackage[round]{natbib}
\usepackage[hyperfootnotes=false]{hyperref}
\hypersetup{
  colorlinks,
  citecolor=Blue,
  linkcolor=Blue,
  urlcolor=Blue}

\numberwithin{equation}{section}

\begin{document}
\begin{bibunit}
\begin{spacing}{1.2}

\title{Tuning-Free Tests for Treatment Effect Heterogeneity\\
via Expectile Processes}

\author{Haokun Lu\thanks{Department of Business Statistics and Econometrics,
Guanghua School of Management, Peking University, Beijing, 100871, China. Email: \texttt{hklu25@stu.pku.edu.cn}.}\\Peking University
\and Xiaojun Song\thanks{Department of Business Statistics and Econometrics, Guanghua School of Management, Peking University, Beijing, 100871, China. Email: \texttt{sxj@gsm.pku.edu.cn}. This work was supported by the National Natural Science Foundation of China [Grant Numbers 72373007, 72333001, and 72621002]. The author also gratefully acknowledges the research support from the Center for Statistical Science of Peking University and the Key Laboratory of Mathematical Economics and Quantitative Finance (Peking University) of the Ministry of Education, China.}\\Peking University}
\maketitle
\vspace{-0.5cm}
\begin{abstract}
This paper develops tuning-free tests for treatment effect heterogeneity. By characterizing the null hypothesis as constancy of the expectile treatment effect curve, we construct test statistics as continuous functionals of the mean-centered sample expectile process. We establish a uniform Bahadur representation, derive the asymptotic properties of the test statistics under the null and fixed alternatives, and show nontrivial power against local alternatives approaching the null at the $n^{-1/2}$ rate. For feasible inference, we develop a fast multiplier bootstrap that consistently approximates the nonpivotal null distribution while capturing the first order effect of estimating the unknown common shift. A key practical advantage of the procedure is that it is tuning-free, requiring no density or score function estimation, bandwidth selection, or repeated refitting. We further introduce pointwise studentization and, for empirical settings in which baseline covariates are available, a covariate adjustment that exploits pretreatment information to improve precision. Monte Carlo comparisons with existing methods show generally reliable size control, competitive power, and substantial computational savings. An application to a large scale water conservation experiment illustrates the proposed tests.

\par\vspace{0.5em}

\end{abstract}
\end{spacing}

\thispagestyle{empty}
\newpage
\normalsize

\section{Introduction}
\label{sec:intro}

Understanding treatment effect heterogeneity has become a central focus in causal inference and program evaluation \citep[see, for example][]{heckman1997making,manski2004statistical,crump2008nonparametric,athey2016recursive,wager2018estimation}. One important reason treatment effect heterogeneity has received widespread attention is that average treatment effects summarize only the average impact of an intervention on an outcome. While an intervention may change the mean outcome, it can also change the dispersion, skewness, or tail behavior of the outcome distribution. These changes in the shape of the distribution are not captured by average treatment effects, yet they may have important implications for inequality and welfare \citep{abadie2002instrumental,bitler2006whatmean}. In this paper, we examine treatment effect heterogeneity from a distributional perspective. Let $F_d$ denote the distribution function of the potential outcome $Y(d)$, where $d\in\{0,1\}$. We consider the following null hypothesis:

\begin{equation}
\label{eq:H0}
H_0:\quad
\exists\,\delta\in\R
\ \text{such that}\
F_1(y)=F_0(y-\delta)
\quad\text{for all }y\in\R.
\end{equation}

The null hypothesis in \eqref{eq:H0} implies that the treatment alters the marginal outcome distribution solely through a common location shift $\delta$, where this shift necessarily equals the average treatment effect. In this distributional sense, the treatment effect is homogeneous: once the average effect is accounted for, the treatment induces no further changes to the shape of the outcome distribution. If $\delta$ were known, testing this null hypothesis would reduce to a classical two-sample problem, which could be addressed using methods such as the Kolmogorov--Smirnov test, the Cram\'er--von Mises test, or a permutation test \citep{lehmann2005testing}. In practice, however, the population average treatment effect is typically unknown; thus, $\delta$ must be estimated from the observed sample. This estimation can alter the first-order behavior of the empirical CDF and quantile processes commonly used in this literature. Consequently, conventional critical values or naive permutation procedures that ignore this estimation effect can lead to invalid inference, even asymptotically. This is the classical Durbin problem \citep{durbin1973weak,koenker2002inference}. The existing literature addresses this issue primarily through two approaches. The first involves using martingale transformations to eliminate the first-order impact of estimating nuisance parameters, thereby rendering the transformed process asymptotically pivotal \citep{khmaladze1981martingale,khmaladze1993goodness}. \citet{koenker2002inference} apply this technique to inference for quantile regression processes. Since quantile functions and distribution functions are equivalent representations of the same underlying distribution, their proposed location-shift test is also applicable to the null hypothesis considered in this paper. The second approach utilizes resampling methods to obtain valid critical values. \citet{chernozhukov2005inference} employ a recentered subsampling method for quantile treatment effect processes, while \citet{ding2016randomization} propose maximizing the $p$-value of a Fisher randomization test over a confidence set for the unknown shift parameter. More recently, \citet{chung2021permutation,chung2025quantile} combine both strategies: they first apply Khmaladze transformations to the distribution and quantile processes, respectively, to construct asymptotically pivotal statistics, and subsequently perform permutation-based inference using these pivotal statistics, drawing on the frameworks established by \citet{neuhaus1993conditional,janssen1997studentized,chung2013exact}. While these methods offer well-developed solutions to the testing problem at hand, each involve different considerations when implemented in practice. Procedures based on Khmaladze transformations typically require the nonparametric estimation of objects such as density functions and score functions; their finite-sample performance can be sensitive to the tuning choice such as bandwidth \citep[see, for example, the simulation results in][]{koenker2002inference}. Alternatively, subsampling methods, which involve selecting a subsample size and performing repeated estimation, can be computationally expensive. These challenges motivate a more direct, computationally efficient approach that does not require additional tuning parameters.

In this paper, we develop a testing procedure based on expectiles, which were introduced by \citet{newey1987asymmetric}. As the name and definition suggest, expectiles are constructed in a manner analogous to quantiles: just as quantiles generalize the median through an asymmetric absolute loss, expectiles generalize the mean through an asymmetric squared loss. Similar to quantiles, the full expectile curve characterizes the entire distribution and exhibits equivariance under location shifts. Consequently, the null hypothesis \eqref{eq:H0} is equivalent to the condition that the difference between the expectile curves of the treatment and control groups remains constant across all indices. We refer to this difference as the expectile treatment effect (ETE) curve. In particular, the value of this curve at $\tau=1/2$ is exactly equal to the average treatment effect. Unlike the methods discussed above, which typically require estimating the common nuisance parameter $\delta$ separately from the distribution function or quantile process and accounting for its estimation effect, our expectile formulation incorporates $\delta$ under the null hypothesis as a specific coordinate of the ETE curve at $\tau=1/2$. The appeal of expectiles is not only that they provide an alternative representation of the distribution, but also that their estimating equations have a structure particularly convenient for statistical inference. While the Jacobian of a quantile estimating equation involves the unknown density at the corresponding quantile, the expectile counterpart depends only on the value of the distribution function at the corresponding expectile \citep{newey1987asymmetric,holzmann2016expectile}. Consequently, all unknown quantities in the first-order representation can be estimated using sample expectiles, along with simple sample averages of weighted residuals and indicator functions. The entire procedure requires no additional estimation of density or auxiliary score functions. It also does not involve nonparametric smoothing or bandwidth selection. This is the core intuition behind our tuning-free inference.

However, the density-free nature of expectiles does not, by itself, resolve the Durbin problem. Although centering the ETE curve at its mean coordinate removes the unknown common shift, the estimation of that shift still affects the sample process at first order.  We establish a uniform Bahadur representation for the centered sample ETE process, which explicitly accounts for this estimation effect, and prove that under the null hypothesis, the test process converges weakly to a Gaussian limit that is generally non-pivotal. Instead of transforming the process to achieve asymptotic pivotality, we retain the estimation effect and use the uniform representation to directly approximate the nonpivotal null distribution via a fast multiplier bootstrap. Throughout the bootstrap procedure, the contributions of the influence function are computed only once based on the original sample. Thereafter, each bootstrap iteration requires only random reweighting, without the need for resampling observations or refitting the expectile curve. We prove the asymptotic validity of the bootstrap and show that the resulting Kolmogorov--Smirnov and Cram\'er--von Mises tests have asymptotically correct size under the null and are consistent against fixed departures from constancy. We also characterize the behavior of the tests under local alternatives that approach the null at the $n^{-1/2}$ rate.

To enhance the practical usefulness of the proposed tests, we introduce two further refinements. First, we use pointwise studentization to account for variations in sampling variance at different expectile levels. This adjustment may increase the test's power, especially when sampling precision varies significantly along the expectile curve, for instance, when the outcome variable exhibits strong skewness or heavy tails. Second, since researchers often have access to extensive baseline data, we introduce a covariate adjustment into the testing framework. Under the null hypothesis, this adjustment reduces the asymptotic covariance kernel of the test process in the positive semidefinite order and thus may also improve power, even if the working regression model is misspecified. Both improvements preserve the tuning-free property of the tests and can be incorporated into the same computationally efficient multiplier bootstrap framework.

The remainder of this paper is organized as follows. Section~\ref{sec:framework} introduces the testing problem and the centered expectile process. Section~\ref{sec:asymptotic} derives the asymptotic distribution under the null and investigates the power properties of the proposed tests under both fixed and local alternatives. Section~\ref{sec:inference} presents the multiplier bootstrap procedure and the pointwise studentized tests, while Section~\ref{sec:covariate-adjustment} details the covariate adjustment method. Section~\ref{sec:simu} employs Monte Carlo simulations to compare the size, power, and computational performance of our tests with those of existing methods across various data generating processes. Section~\ref{sec:empirical} applies the proposed approach to the large-scale water conservation experiment of \citet{ferraro2013using} and examines treatment effect heterogeneity across different interventions. Section~\ref{sec:con} concludes the paper. All proofs and additional simulation results are provided in the Online Appendix. An R package implementing the proposed tests is publicly available at \url{https://github.com/HaokunLU722/embt}.

\textsc{Notation}. For an index set $\mathcal A$, let $\ell^\infty(\mathcal A)$ denote the Banach space of bounded real-valued functions on $\mathcal A$, equipped with the supremum norm $\|f\|_\infty=\sup_{a\in\mathcal A}|f(a)|$. We write $X_n\rightsquigarrow X$ for weak convergence of random elements in $\ell^\infty(\mathcal A)$ in the Hoffmann--J{\o}rgensen sense; see, for example, Definition 1.3.3 of \citet{vdvw2023}. We write $\xrightarrow{p}$ for convergence in probability. The norm $\|\cdot\|$ is Euclidean for finite-dimensional vectors, and stochastic order statements for random functions are understood with respect to the supremum norm.

\section{Test Formulation}
\label{sec:framework}

\subsection{Setup}
\label{subsec:framework-null}

Consider a randomized experiment with a binary treatment indicator $D_i\in\{0,1\}$ and potential outcomes $Y_i(1)$ and $Y_i(0)$. The observed outcome is $Y_i=D_iY_i(1)+(1-D_i)Y_i(0)$. Let $F_d$ denote the marginal distribution function of $Y_i(d)$, where $d\in\{0,1\}$. Under random assignment independent of the potential outcomes, these distributions are identified from the observed outcomes in the respective treatment arms. The null hypothesis considered in this paper restricts the marginal distributions of the potential outcomes. To clarify its relationship with individual-level treatment effect homogeneity, consider a stronger null hypothesis
\begin{equation}
\label{eq:H0S}
H_0^S:\qquad
Y_i(1)-Y_i(0)=\delta
\quad\text{for all }i
\quad\text{for some }\delta\in\mathbb R.
\end{equation}
Under $H_0^S$, the two marginal distributions differ only by a common location shift $\delta$, so $H_0^S$ implies $H_0$. The converse need not hold because $H_0$ restricts only the marginal distributions of $Y(1)$ and $Y(0)$, not their joint distribution. Furthermore, since only one potential outcome is observed for each unit, this joint distribution cannot be identified from the observed data, which is the fundamental problem of causal inference \citep{holland1986statistics}. Thus, rejecting $H_0$ also rules out constant individual treatment effects, but failure to reject $H_0$ does not imply that individual treatment effects are homogeneous. This distinction is also reflected in \citet{ding2016randomization}: although their randomization test is formulated for the sharp individual effect null $H_0^S$, they also use departures of the marginal distributions from a common location shift as evidence of treatment effect heterogeneity.

For the asymptotic analysis, we adopt the following independent two-sample superpopulation framework:
\begin{equation}\label{eq:samples}
Y_1^{(1)},\ldots,Y_{n_1}^{(1)} \stackrel{\mathrm{iid}}{\sim}F_1,
\qquad
Y_1^{(0)},\ldots,Y_{n_0}^{(0)} \stackrel{\mathrm{iid}}{\sim}F_0,
\end{equation}
where the two samples are mutually independent.

\subsection{Expectile Formulation of the Null}
\label{subsec:expectiles}

We now reformulate the null hypothesis \eqref{eq:H0} in terms of expectiles. For $\tau\in(0,1)$ and a random variable $Y$ with finite second moment, the $\tau$-expectile is the unique minimizer of the asymmetric least-squares criterion
\begin{equation}
\label{eq:expectile-loss}
e(\tau)
=
\argmin_{e\in\R}
\E\!\left[
\left|\tau-\1\{Y\le e\}\right|(Y-e)^2
\right].
\end{equation}
The index $\tau$ determines the relative weights placed on positive and negative residuals. When $\tau>1/2$, observations above $e$ receive greater weight, whereas when $\tau<1/2$, observations below $e$ receive greater weight. At $\tau=1/2$, the criterion reduces to ordinary squared loss, and hence $e(1/2)=\E Y$ \citep{newey1987asymmetric}. The corresponding first-order condition is
\begin{equation}
\label{eq:psi-def}
\E\psi_\tau\{Y,e(\tau)\}=0,
\qquad
\psi_\tau(y,e)
=
2\abs{\tau-\1\{y\le e\}}(y-e).
\end{equation}
Expectiles are translation equivariant: for any $b\in\R$, $e_{Y+b}(\tau)=e_Y(\tau)+b$.

Let $e_d(\tau)$ denote the $\tau$-expectile of $Y(d)$ for $d\in\{0,1\}$. Following \citet{stahlschmidt2014expectile}, define the \emph{expectile treatment effect} (ETE) curve as
\begin{equation}
\label{eq:ETE-def}
\Delta(\tau)
=
e_1(\tau)-e_0(\tau),
\qquad
\tau\in(0,1).
\end{equation}
The following result gives the expectile characterization of the null.
\begin{proposition}
\label{prop:full-curve}
Suppose that $F_0$ and $F_1$ have finite first moments. Then
\[
H_0
\quad\Longleftrightarrow\quad
\exists\,\delta\in\R
\ \text{such that}\
\Delta(\tau)=\delta
\quad\text{for all }\tau\in(0,1).
\]
\end{proposition}
Let $\mu_d=\E\{Y(d)\}$. Since $e_d(1/2)=\mu_d$, the constant ETE under $H_0$ equals $\Delta(1/2)=\mu_1-\mu_0$. We thus center the ETE curve at its $1/2$ coordinate and define:
\begin{equation}
\label{eq:r-def}
r(\tau)
=
\Delta(\tau)-\Delta(1/2)
=
e_1(\tau)-e_0(\tau)-(\mu_1-\mu_0).
\end{equation}
Under $H_0$, $r(\tau)=0$ for every $\tau\in(0,1)$. This motivates a process-based test of the null. We work with a fixed closed interior interval $\mathcal T=[\tau_L,\tau_U]\subset(0,1)$ containing $1/2$, and base inference on global functionals of a suitably normalized sample analog of $r(\tau)$ over $\mathcal T$. 

\subsection{The Expectile Process and Test Statistics}
\label{subsec:sample-process}

For $d\in\{0,1\}$, the sample $\tau$-expectile $\widehat e_d(\tau)$ is defined as
\begin{equation}
\label{eq:sample-expectile}
\widehat e_d(\tau)
=
\argmin_{e\in\R}
\frac{1}{n_d}
\sum_{i=1}^{n_d}
\left|\tau-\1\{Y_i^{(d)}\le e\}\right|
\left(Y_i^{(d)}-e\right)^2,
\end{equation}
which also satisfies
$n_d^{-1}\sum_{i=1}^{n_d}\psi_\tau\{Y_i^{(d)},\widehat e_d(\tau)\}=0$.
Define the sample ETE curve as
$\widehat\Delta(\tau)=\widehat e_1(\tau)-\widehat e_0(\tau)$,
and let $\bar Y_d=n_d^{-1}\sum_{i=1}^{n_d}Y_i^{(d)}$.
Since $\widehat e_d(1/2)=\bar Y_d$, the sample analogue of the centered ETE curve is
$\widehat r(\tau)=\widehat\Delta(\tau)-(\bar Y_1-\bar Y_0)$. Write $n=n_1+n_0$, $\lambda_n=n_1/n$, and  $N=n_1n_0/n$. We scale $\widehat r$ by the usual two-sample rate $\sqrt{N}$ to obtain the centered ETE process
\begin{equation}
\label{eq:S-def}
S_n(\tau)
=
\sqrt{N}\,\widehat r(\tau),
\qquad
\tau\in\T.
\end{equation}
We measure departures from zero over $\T$ with the Kolmogorov--Smirnov and Cram\'er--von Mises test statistics
\begin{equation}
\label{eq:test-stats}
KS_n
=
\sup_{\tau\in\T}|S_n(\tau)|,
\qquad
CvM_n
=
\int_{\T}S_n^2(\tau)\,d\tau.
\end{equation}
Both tests reject the null hypothesis for sufficiently large values of their respective statistics. The next section develops the asymptotic theory underlying these statistics.

\section{Asymptotic Theory}
\label{sec:asymptotic}

For arm $d\in\{0,1\}$, let $\E_d$ denote expectation under $F_d$, let $\E_{n,d}f=n_d^{-1}\sum_{i=1}^{n_d}f(Y_i^{(d)})$, and define $\mathbb G_{n,d}f=\sqrt{n_d}(\E_{n,d}-\E_d)f$. We impose the following assumptions.

\begin{assumption}
\label{ass:sampling}
The samples in \eqref{eq:samples} are independent, observations are i.i.d. within each arm, $n_1\to\infty$, $n_0\to\infty$, and $\lambda_n\to\lambda\in(0,1)$.
\end{assumption}

\begin{assumption}
\label{ass:regularity}
For each $d\in\{0,1\}$, $\E_d Y^2<\infty$. The CDF $F_d$ is continuous on an open neighborhood of $\mathcal E_d(\T)=\{e_d(\tau):\tau\in\T\}$.
\end{assumption}

Assumption~\ref{ass:sampling} is standard for two-sample asymptotic analysis. Assumption~\ref{ass:regularity} reflects the main regularity trade-off in expectile-based inference. An expectile is a moment-based functional whose definition requires the random variable to have a finite first moment, while the second-moment condition supports Gaussian inference at the root-$n$ convergence rate. This is the price paid relative to quantile methods: the distributional functionals corresponding to quantiles remain well defined in the absence of these moment conditions and are therefore more robust to heavy-tailed distributions \citep[see][]{newey1987asymmetric}. Apart from this moment requirement, however, we require only that the cumulative distribution function be continuous in a neighborhood of the relevant expectile range. This condition is sufficient for the uniform first-order linearization used below; see \citet{holzmann2016expectile}. We do not require a density to exist, let alone to be estimated. This feature, as emphasized by \citet{newey1987asymmetric}, makes expectile-based inference particularly convenient for empirical implementation because the first-order covariance structure can be estimated directly from sample moments. Compared with the existing literature, the CDF-based procedure of \citet{chung2021permutation} relies on smooth absolute continuity conditions and feasible estimation of density and score objects for its martingale transformation. Likewise, the quantile-based procedure of \citet{chung2025quantile} requires positive densities at the relevant quantiles and uniformly consistent estimation of density and score functions. Thus, the expectile formulation trades stronger moment requirements for tuning-free inference. In other words, it requires no smoothing-based, nonparametric nuisance estimation, such as density, sparsity, or score estimation. Consequently, no associated tuning parameter needs to be selected.

To understand why density estimation is unnecessary, consider the Jacobian of the expectile score $\psi_\tau$ in \eqref{eq:psi-def}. For each arm,
\begin{equation}
\label{eq:score_jacobian}
D_d(\tau)
=
-\left.
\frac{\partial}{\partial e}
\E_d\psi_\tau(Y,e)
\right|_{e=e_d(\tau)}
=
2\left[
\tau+(1-2\tau)F_d\{e_d(\tau)\}
\right]
=
2\E_d\abs{\tau-\1\{Y\le e_d(\tau)\}}.
\end{equation}

The influence function of the $\tau$-expectile is then given by
\begin{equation}
\label{eq:phi-def}
\phi_{d,\tau}(y)
=
\frac{\psi_\tau\{y,e_d(\tau)\}}{D_d(\tau)},
\qquad
\E_d\phi_{d,\tau}
=
0.
\end{equation}
On the fixed interior interval $\T=[\tau_L,\tau_U]$, $\inf_{\tau\in\T}D_d(\tau)
\ge
2\min\{\tau_L,1-\tau_U\}
>
0$. The denominator in \eqref{eq:phi-def} is therefore uniformly bounded away from zero, so the influence function is well-defined uniformly over $\T$. The following lemma provides a uniform Bahadur representation of the sample expectile curve.

\begin{lemma}
\label{lem:bahadur}
Under Assumption~\ref{ass:regularity}, for each $d\in\{0,1\}$,
\begin{equation}
\label{eq:bahadur}
\widehat e_d(\tau)-e_d(\tau)
=
\E_{n,d}\phi_{d,\tau}
+
R_{d,n_d}(\tau),
\qquad
\sup_{\tau\in\T}
\abs{R_{d,n_d}(\tau)}
=
o_p(n_d^{-1/2}).
\end{equation}
Moreover, the function class $\{\phi_{d,\tau}:\tau\in\T\}$ is Donsker under $F_d$ with an envelope bounded by $C(1+|Y|)$, for some constant $C$.
\end{lemma}

The remainder $o_p(n_d^{-1/2})$ in Lemma~\ref{lem:bahadur} is already sufficient for the first-order asymptotic analysis below. Furthermore, if the distribution function $F_d$ admits a bounded density in a neighborhood of the expectile range $\mathcal E_d(\T)$, then the remainder can be sharpened to $O_p(n_d^{-1})$ uniformly over the interval. Intuitively, the bounded density condition makes the CDF locally Lipschitz, so the nonlinear component of the score expansion is of second order in the expectile estimation error. This sharper remainder rate may be of independent interest. Nevertheless, our subsequent results do not require this stronger condition.

\subsection{Null Asymptotics}
\label{subsec:asymptotic-limit}

Since $\bar Y_d-\mu_d=\E_{n,d}(Y-\mu_d)$, the influence function of the mean-centered expectile $e_d(\tau)-\mu_d$ is
\begin{equation}\label{eq:xi-def}
\xi_{d,\tau}(y)
=
\phi_{d,\tau}(y)-(y-\mu_d).
\end{equation}
In particular, $\E_d\xi_{d,\tau}=0$ and $\xi_{d,1/2}\equiv0$.

\begin{theorem}
\label{thm:centered-limit}
Under Assumptions~\ref{ass:sampling}--\ref{ass:regularity}, uniformly in $\tau\in\T$,
\begin{equation}\label{eq:S-linear}
S_n(\tau)
=
\sqrt{N}r(\tau)
+
\sqrt{1-\lambda_n}\,\mathbb G_{n,1}\xi_{1,\tau}
-
\sqrt{\lambda_n}\,\mathbb G_{n,0}\xi_{0,\tau}
+o_p(1).
\end{equation}
Under $H_0$,
\begin{equation}\label{eq:S-null-limit}
S_n\rightsquigarrow\mathbb G_{\xi}
\qquad\text{in }\ell^\infty(\T),
\end{equation}
where $\mathbb G_{\xi}$ is a centered Gaussian process with covariance kernel $\Gamma(\tau,\tau')=\E_d\{\xi_{d,\tau}\xi_{d,\tau'}\}$, $d\in\{0,1\}$.
\end{theorem}

Under $H_0$, translation equivariance implies that the two arm-specific influence function processes have the same distribution. Their covariance kernels therefore coincide, and the weights $1-\lambda$ and $\lambda$ sum to one. Hence, the null limit does not depend on the treatment probability. The limiting distributions of the KS and CvM statistics then follow by the continuous mapping theorem.

\begin{corollary}\label{cor:functionals}
Under the conditions of Theorem~\ref{thm:centered-limit} and $H_0$,
\begin{equation*}
KS_n
\rightsquigarrow
\sup_{\tau\in\T}|\mathbb G_{\xi}(\tau)|,
\qquad
CvM_n
\rightsquigarrow
\int_\T\mathbb G_{\xi}^2(\tau)\,d\tau.
\end{equation*}
\end{corollary}

\begin{remark}[The Durbin problem]
Mean centering removes the unknown location shift from the population restriction, but its estimation still affects the limiting process. Under $H_0$, the common shift $\delta=\mu_1-\mu_0$ is estimated by $\bar Y_1-\bar Y_0$, which give rise to the terms $-(y-\mu_d)$ in \eqref{eq:xi-def}. If $\delta$ were known, these terms would diminish. However, the resulting Gaussian limit would nevertheless remain nonpivotal, since its covariance kernel still depends on the underlying outcome distribution. Thus, the effect of estimating $\delta$ and the nonpivotality of the limiting process are distinct issues.
\end{remark}

\subsection{Asymptotic Power}
\label{subsec:asymptotic-power}
We now investigate the power properties of the centered process tests, first under fixed alternatives and then under local departures from the null.  Under fixed alternatives, the uniform representation in Theorem~\ref{thm:centered-limit} implies that rescaled process converges uniformly in probability to the population discrepancy $r$.

\begin{proposition}
\label{prop:fixed}
Suppose Assumptions~\ref{ass:sampling}--\ref{ass:regularity} hold.  Then
\[
\frac{KS_n}{\sqrt N}
\xrightarrow p
\sup_{\tau\in\T}|r(\tau)|,
\qquad
\frac{CvM_n}{N}
\xrightarrow p
\int_\T r^2(\tau)\,d\tau.
\]
In particular, if $r\not\equiv0$ on $\T$, both statistics diverge in probability.
\end{proposition}

Thus, the tests are consistent against any fixed alternative $H_1$ for which $r\not\equiv0$ on $\T$.  We next turn to Pitman-type local alternatives that approach the null at the parametric rate $N^{-1/2}$.  Let $\varepsilon_0$ and $\varepsilon_1$ be independent copies of a mean-zero random variable $\varepsilon\sim F$, we first characterize the null data generating process as
\[
Y^{(0)}=\mu+\varepsilon_0,
\qquad
Y^{(1)}=\mu+\delta+\varepsilon_1.
\]
To construct local alternatives while keeping $\delta$ fixed, we perturb the treated innovation $\varepsilon_1$. Let $b:\mathbb R\to\mathbb R$ be a nondecreasing and globally Lipschitz continuous function. For $s$ near zero, define the transformation $T_s(x)=x+s b(x)$. We impose $\E b(\varepsilon)=0$, so that $\E T_s(\varepsilon)=0$ for every $s$. We further assume $\Var\{b(\varepsilon)\}>0$ to exclude the trivial perturbation. For fixed $\kappa\ne0$, we set $s_n=\kappa/\sqrt N$ and consider the following local alternatives converging to the null at $N^{-1/2}$ rate:
\begin{equation}\label{eq:local-outcome-path}
H_{1n}(b,\kappa):
\qquad
Y_n^{(0)}\stackrel d=\mu+\varepsilon_0,
\qquad
Y_n^{(1)}\stackrel d=
\mu+\delta+T_{s_n}(\varepsilon_1).
\end{equation}

To investigate local power, we first examine the first-order effect of the perturbation on the expectile curve. Let $e_s(\tau)$ denote the $\tau$-expectile of $T_s(\varepsilon)$, with $e_0(\tau)=e(\tau)$, and let $D(\tau)$ denote the score Jacobian of the null distribution defined in \eqref{eq:score_jacobian}. Differentiating the population expectile first-order condition with respect to $s$ gives
\begin{equation}\label{eq:hb-def}
h_b(\tau)
:=
\left.
\frac{\partial e_s(\tau)}{\partial s}
\right|_{s=0}
=
\frac{
2\E\!\left[
\abs{\tau-\1\{\varepsilon\le e(\tau)\}}b(\varepsilon)
\right]
}{
D(\tau)
}.
\end{equation}
As $s\to0$, uniformly over $\T$, $e_s(\tau)=e(\tau)+s h_b(\tau)+o(|s|)$. Since $\E T_s(\varepsilon)=0$, the mean difference between treatment and control remains $\delta$. Thus, under $H_{1n}(b,\kappa)$, the centered population ETE curve is $r_n(\tau)=e_{s_n}(\tau)-e(\tau)$.
\begin{theorem}
\label{thm:local}
Suppose Assumptions~\ref{ass:sampling} and \ref{ass:regularity} hold. Under $H_{1n}(b,\kappa)$,
\begin{equation}\label{eq:path-drift}
\sup_{\tau\in\T}
\left|\sqrt N\,r_n(\tau)-\kappa h_b(\tau)\right|
\to0,
\end{equation}
and
\begin{equation}\label{eq:S-local}
S_n\rightsquigarrow\mathbb G_\xi+\kappa h_b
\qquad\text{in }\ell^\infty(\T),
\end{equation}
where $\mathbb G_\xi$ is the null Gaussian process in Theorem~\ref{thm:centered-limit}. 
\end{theorem}
 By the continuous mapping theorem,
\[
KS_n
\rightsquigarrow
\sup_{\tau\in\T}
\abs{\mathbb G_\xi(\tau)+\kappa h_b(\tau)},
\qquad
CvM_n
\rightsquigarrow
\int_\T
\{\mathbb G_\xi(\tau)+\kappa h_b(\tau)\}^2\,d\tau.
\]
Hence, under local alternatives approaching the null at the parametric rate $N^{-1/2}$, both statistics have nondegenerate shifted limits, yielding nontrivial local power against $H_{1n}(b,\kappa)$.  The following scale perturbation provides a simple illustration of the resulting drift.

\begin{example}[Local scale perturbation]
\label{ex:local-scale-path}
Consider $b(x)=x$, so that $T_{s_n}(x)=(1+\kappa/\sqrt{N})x$. The local alternative then takes the familiar location-scale form $Y_n^{(0)}\sim\mu+\varepsilon_0$ and $Y_n^{(1)}\sim\mu+\delta+(1+\kappa/\sqrt{N})\varepsilon_1$. In this model, treatment induces a location shift of size $\delta$ and a scale change of order $N^{-1/2}$. For all sufficiently large $n$, $\Delta_{E,n}(\tau)=\delta+\kappa e(\tau)/\sqrt{N}$ and $\sqrt{N}r_n(\tau)=\kappa e(\tau)$. Thus, $h_b(\tau)=e(\tau)$, and the local drift is proportional to the expectile curve of the underlying distribution. For $\kappa>0$, treatment increases dispersion, moves lower expectiles downward and upper expectiles upward relative to $\delta$. These directions are reversed when $\kappa<0$.
\end{example}

\section{Tuning-Free Inference}
\label{sec:inference}

The null limiting distribution in Theorem~\ref{thm:centered-limit} is generally not pivotal. Nevertheless, the uniform Bahadur representation in Lemma~\ref{lem:bahadur} offers a direct path toward feasible inference, as the relevant components of the influence function can be estimated using simple sample analogues. To obtain feasible critical values, this section develops a computationally efficient multiplier bootstrap procedure that approximates the unknown null distribution by randomly reweighting the estimated components of the influence function. Within the same framework, we also introduce a pointwise studentized refinement.

\subsection{Multiplier Bootstrap}
\label{subsec:multiplier}

We first construct feasible counterparts of the influence functions in Section~\ref{sec:asymptotic}. For arm $d\in\{0,1\}$, let $\widehat F_d$ denote the empirical CDF and define
\begin{align}
\widehat D_d(\tau)
&=
2\left[
\tau+(1-2\tau)\widehat F_d\{\widehat e_d(\tau)\}
\right],
\label{eq:Dhat}\\
\widehat\xi_{d,i,\tau}
&=
\frac{
\psi_\tau\{Y_i^{(d)},\widehat e_d(\tau)\}
}{
\widehat D_d(\tau)
}
-
\{Y_i^{(d)}-\bar Y_d\}.
\label{eq:xihat}
\end{align}
Thus, $\widehat\xi_{d,i,\tau}$ is the direct plug-in analogue of $\xi_{d,\tau}$ in \eqref{eq:xi-def}. In particular, $\E_{n,d}\widehat\xi_{d,\tau}=0$ for every $\tau\in\T$, and $\widehat\xi_{d,i,1/2}=0$.

For each arm, let $\{V_i^{(d)}\}_{i=1}^{n_d}$ be independent and identically distributed multipliers with mean zero, unit variance, and bounded support. Suppose that the multipliers are independent across arms and independent of the data. A convenient choice is the Bernoulli random variable $V$ with $P(V=1-a)=a/\sqrt{5}$ and $P(V=a)=1-a/\sqrt{5}$, where $a=(1+\sqrt{5})/2$, see \cite{mammen1993bootstrap}. Using the feasible influence values defined in \eqref{eq:xihat}, define the multiplier process
\begin{equation}
\label{eq:Sstar-def}
S_n^*(\tau)
=
\sqrt{1-\lambda_n}\frac{1}{\sqrt{n_1}}
\sum_{i=1}^{n_1}
V_i^{(1)}\widehat\xi_{1,i,\tau}
-
\sqrt{\lambda_n}\frac{1}{\sqrt{n_0}}
\sum_{j=1}^{n_0}
V_j^{(0)}\widehat\xi_{0,j,\tau},
\qquad
\tau\in\T.
\end{equation}

For each draw, the bootstrap analogues of the two test statistics are computed as:
\begin{equation}
\label{eq:bootstrap-stats}
KS_n^*
=
\sup_{\tau\in\T}|S_n^*(\tau)|,
\qquad
CvM_n^*
=
\int_{\T}\{S_n^*(\tau)\}^2\,d\tau.
\end{equation}
Repeating the bootstrap draw $B$ times yields
$\{KS_{n,b}^*,CvM_{n,b}^*\}_{b=1}^B$.
For a significance level $\alpha$, let $\widehat c_{KS,1-\alpha}^*$ and $\widehat c_{CvM,1-\alpha}^*$ denote the corresponding empirical $(1-\alpha)$-quantiles. The KS and CvM tests then reject the null hypothesis when $KS_n>\widehat c_{KS,1-\alpha}^*$ and $CvM_n>\widehat c_{CvM,1-\alpha}^*$, respectively. Since all quantities involved in $\widehat\xi_{d,i,\tau}$ are computed only once from the original sample and remain fixed across bootstrap draws, the entire bootstrap procedure requires only repeated reweighting of the precomputed influence values in \eqref{eq:Sstar-def}. Specifically, there is no need to re-estimate expectiles or other auxiliary quantities for each draw. Write $\rightsquigarrow^*$ for weak convergence in probability under the bootstrap law, i.e., conditional on the original sample $\{Y_i^{(1)}\}_{i=1}^{n_1}\cup\{Y_j^{(0)}\}_{j=1}^{n_0}$,  The following theorem establishes the validity of this bootstrap procedure.

\begin{theorem}
\label{thm:multiplier-fclt}
Suppose Assumptions~\ref{ass:sampling} and \ref{ass:regularity} hold. Under $H_0$ or the sequence of local alternatives $H_{1n}(b,\kappa)$,
\begin{equation*}
S_n^*
\rightsquigarrow^*
\mathbb G_\xi
\quad
\text{in }\ell^\infty(\T),
\end{equation*}
where $\mathbb G_\xi$ is the centered Gaussian process in Theorem~\ref{thm:centered-limit}.

Under the fixed alternative $H_1$,
\begin{equation*}
S_n^*
\rightsquigarrow^*
\mathbb G_\xi^{A}
\quad
\text{in }\ell^\infty(\T),
\end{equation*}
where $\mathbb G_\xi^{A}$ is a centered Gaussian process with covariance kernel
\begin{equation*}
\Gamma_A(\tau,\tau')
=
(1-\lambda)\E_1\{\xi_{1,\tau}\xi_{1,\tau'}\}
+
\lambda\E_0\{\xi_{0,\tau}\xi_{0,\tau'}\}.
\end{equation*}
\end{theorem}

Under $H_0$, Theorem~\ref{thm:multiplier-fclt} shows that, conditional on the original sample, the multiplier process converges weakly to the same Gaussian process $\mathbb G_\xi$ as the observed process in Theorem~\ref{thm:centered-limit}. Therefore, the tests based on the bootstrap critical values have asymptotically correct size. Under the local alternatives $H_{1n}(b,\kappa)$, the multiplier process continues to reproduce the null limit, while Theorem~\ref{thm:local} shows that the observed process acquires the deterministic drift $\kappa h_b$. Hence, the bootstrap procedure preserves the nontrivial local power. Under any fixed alternative in $H_1$, the second part of Theorem~\ref{thm:multiplier-fclt} implies that the bootstrap statistics, and hence their critical values, remain stochastically bounded, while Proposition~\ref{prop:fixed} shows that the observed test statistic diverge. This implies that the multiplier bootstrap tests are consistent against any fixed alternative $H_1$ for which $r\not\equiv0$ on $\T$.

\subsection{A Studentized Refinement}
\label{subsec:studentization}

The multiplier bootstrap developed above provides feasible inference for the KS and CvM tests. However, the sampling variance of $S_n(\tau)$ may vary across $\T$, so coordinates with larger variances may have a greater influence on the test statistics, even when they carry little signal under the alternative. We therefore consider studentizing the process by its pointwise standard deviation. A similar approach is used by \citet{chernozhukov2005inference}, who consider inverse pointwise covariance weighting for quantile process tests and refer to it as Anderson--Darling-type weighting.\footnote{In the classical one-sample goodness-of-fit problem, the Anderson--Darling criterion weights the squared Brownian bridge process by $[t(1-t)]^{-1}$, the inverse of its pointwise variance. Our studentized CvM statistic follows the same principle, while the studentized KS statistic uses pointwise standard error normalization inside the supremum.} The following toy example illustrates the effect on power.

\begin{example}[Variance-dominating noise]
\label{ex:variance-dominance}
Consider testing the hypothesis $H_0:\Delta=0$ against the hypothesis $H_1:\Delta\ne0$ based on two observations, $X_1=Z_1+\Delta$ and $X_2=\rho Z_2$, where $Z_1$ and $Z_2$ are independent standard normal random variables and $\rho>0$. The first coordinate contains the signal, while the second contains only noise with variance $\rho^2$. Consider the unstandardized and standardized KS-type statistics:
\[
T_\rho=\max\{|X_1|,|X_2|\},
\qquad
T_S=\max\{|X_1|,|X_2|/\rho\},
\]
where each uses its own critical value at level $\alpha$ under $H_0$. As $\rho\to\infty$, the noise in the second coordinate dominates $T_\rho$, and its rejection probability approaches $\alpha$ even when $\Delta\ne0$. By contrast, the power of $T_S$ is independent of $\rho$ and remains strictly above $\alpha$. The same result holds for the corresponding CvM-type statistics, $X_1^2+X_2^2$ and $X_1^2+(X_2/\rho)^2$.
\end{example}

Intuitively, this example replaces a continuum of expectile coordinates with just two points. Pointwise studentization can be useful when coordinates with greater sampling variability do not carry correspondingly stronger signals. For the expectile process, this studentization is particularly easy-to-implement because the pointwise variances can be estimated directly from sample moments and empirical CDF values, without the need for nonparametric estimation. However, a specific issue arises when $S_n(\tau)$ is studentized directly. By definition, $S_n(1/2) \equiv 0$, and its pointwise variance is zero at $\tau=1/2$. Consequently, direct division by the pointwise standard deviation produces an 0/0 singularity at this coordinate. To address this issue, we show an exact factorization of the mean-centered expectile.

Define $c_\tau=2\tau-1$ and $m_d(\tau)=\E_d|Y-e_d(\tau)|$. The expectile first-order condition yields the identity $e_d(\tau)-\mu_d=c_\tau m_d(\tau)$. Likewise, the sample expectile satisfies $\widehat e_d(\tau)-\bar Y_d=c_\tau\widehat m_d(\tau)$, where $\widehat m_d(\tau)=\E_{n,d}|Y-\widehat e_d(\tau)|$. Define $q(\tau)=m_1(\tau)-m_0(\tau)$ and $A_n(\tau)=\sqrt N\{\widehat m_1(\tau)-\widehat m_0(\tau)\}$. The same factorization extends to the first-order representation.

\begin{proposition}
\label{prop:factorization}
Let $a_d(\tau)=1-2F_d\{e_d(\tau)\}$ and define
\[
\chi_{d,\tau}(y)
=
\frac{
|y-e_d(\tau)|
-
m_d(\tau)
-
a_d(\tau)(y-\mu_d)
}{
D_d(\tau)
}.
\]
Under Assumption~\ref{ass:regularity}, uniformly over $\tau\in\T$,
\begin{equation}
\label{eq:m-linearization}
\widehat m_d(\tau)-m_d(\tau)
=
\E_{n,d}\chi_{d,\tau}
+
o_p(n_d^{-1/2}).
\end{equation}
For its feasible counterpart, let $\widehat a_d(\tau)= 1-2\widehat F_d\{\widehat e_d(\tau)\}$ and define
\[
\widehat\chi_{d,i,\tau}
=
\frac{
|Y_i^{(d)}-\widehat e_d(\tau)|
-
\widehat m_d(\tau)
-
\widehat a_d(\tau)(Y_i^{(d)}-\bar Y_d)
}{
\widehat D_d(\tau)
}.
\]
Then $\xi_{d,\tau}=c_\tau\chi_{d,\tau}$, $\widehat\xi_{d,i,\tau}=c_\tau\widehat\chi_{d,i,\tau}$, with $\E_{n,d}\widehat\chi_{d,\tau}=0$. Consequently, $r(\tau)=c_\tau q(\tau)$, $S_n(\tau)=c_\tau A_n(\tau)$.
\end{proposition}

The decomposition shows that the singularity at the mean coordinate arises from the known factor $c_\tau$. We therefore studentize $A_n$ directly. Its asymptotic pointwise variance and sample analogue are given by
\begin{equation}\label{eq:student-variance}
v_A^2(\tau)
=
(1-\lambda)\E_1\chi_{1,\tau}^2
+
\lambda\E_0\chi_{0,\tau}^2,
\qquad
\widehat v_{A,n}^2(\tau)
=
(1-\lambda_n)\E_{n,1}\widehat\chi_{1,\tau}^2
+
\lambda_n\E_{n,0}\widehat\chi_{0,\tau}^2.
\end{equation}
The studentized process and the corresponding test statistics are then defined as
\begin{equation}
\label{eq:U-def}
U_n(\tau)
=
\frac{A_n(\tau)}{\widehat v_{A,n}(\tau)},
\qquad
KS_n^{\mathrm{stud}}
=
\sup_{\tau\in\T}|U_n(\tau)|,
\qquad
CvM_n^{\mathrm{stud}}
=
\int_{\T}U_n^2(\tau)\,d\tau.
\end{equation}
For $\tau\ne1/2$, $U_n(\tau)$ agrees with direct studentization of $S_n(\tau)$ up to sign, while remaining well defined at $\tau=1/2$. Under Assumptions~\ref{ass:sampling} and \ref{ass:regularity}, $\widehat v_{A,n}^2$ is uniformly consistent for $v_A^2$, which is finite and continuous on $\T$. For studentization, we additionally require $\inf_{\tau\in\T}v_A^2(\tau)>0$.

\begin{corollary}
\label{cor:studentized}
Suppose Assumptions~\ref{ass:sampling} and \ref{ass:regularity} hold and $\inf_{\tau\in\T}v_A^2(\tau)>0$.  Under $H_0$,
\begin{equation}\label{eq:U-null-limit}
U_n
\rightsquigarrow
\mathbb U
:=
\frac{\mathbb G_\chi}{v_A}
\qquad
\text{in }\ell^\infty(\T),
\end{equation}
where $\mathbb G_\chi$ is a centered Gaussian process with covariance kernel $K(\tau,\tau') = \E_d\{\chi_{d,\tau}\chi_{d,\tau'}\}$, $d\in\{0,1\}$. Consequently, $KS_n^{\mathrm{stud}} \rightsquigarrow \sup_{\tau\in\T}|\mathbb U(\tau)|$ and $CvM_n^{\mathrm{stud}} \rightsquigarrow \int_\T \mathbb U^2(\tau)\,d\tau$.
\end{corollary}

The studentized tests can be implemented using the multiplier bootstrap described in Section~\ref{subsec:multiplier}. Using the same random weights, define
\begin{equation}\label{eq:student-multiplier}
U_n^*(\tau)=\frac{1}{\widehat v_{A,n}(\tau)}\left\{\sqrt{1-\lambda_n}\frac{1}{\sqrt{n_1}}\sum_{i=1}^{n_1}V_i^{(1)}\widehat\chi_{1,i,\tau}-\sqrt{\lambda_n}\frac{1}{\sqrt{n_0}}\sum_{j=1}^{n_0}V_j^{(0)}\widehat\chi_{0,j,\tau}\right\},\qquad \tau\in\T.
\end{equation}
Under $H_0$, the same bootstrap argument and uniform consistency of $\widehat v_{A,n}$ yield $U_n^*\rightsquigarrow^*\mathbb U$ in probability. This ensures asymptotically valid bootstrap critical values for the studentized tests.

We further examine the effect of studentization on local power. Following the local alternatives introduced in Section~\ref{subsec:asymptotic-power}, let $e_s(\tau)$ denote the expectile of $T_s(\varepsilon)$ and define the corresponding absolute-deviation curve as $m_s(\tau)=\E|T_s(\varepsilon)-e_s(\tau)|$, where $m_0(\tau)=m(\tau)$. Differentiating with respect to $s$ gives
\begin{equation}
\label{eq:gb-def}
g_b(\tau)
:=
\left.
\frac{\partial m_s(\tau)}{\partial s}
\right|_{s=0}
=
-\frac{
2\E[b(\varepsilon)\1\{\varepsilon\le e(\tau)\}]
}{
D(\tau)
}.
\end{equation}
Since $e_s(\tau)-\E T_s(\varepsilon)=c_\tau m_s(\tau)$ and $\E b(\varepsilon)=0$, we have $h_b(\tau)=c_\tau g_b(\tau)$. Consequently, under the local alternatives $H_{1n}(b,\kappa)$,
\[
U_n
\rightsquigarrow
\frac{\mathbb G_\chi+\kappa g_b}{v_A}
\qquad
\text{in }\ell^\infty(\T).
\]

Compared with the original tests, studentization does not alter the detection boundary against local alternatives, which remains of order $N^{-1/2}$, but changes the relative weights assigned to signals across expectile levels. The deterministic drift of the original process is $\kappa c_\tau g_b(\tau)$, whereas that of the studentized process is $\kappa g_b(\tau)/v_A(\tau)$. For the local scale alternatives in Example~\ref{ex:local-scale-path}, direct calculation shows that the drift terms for the original and studentized processes are $\kappa c_\tau m(\tau)$ and $\kappa m(\tau)/v_A(\tau)$, respectively. This difference shows that studentization does more than rescale the numerical values of the statistics; it changes the relative influence of different expectile levels on the global test. Studentization can improve power when regions with high sampling variance do not provide sufficiently strong signals. However, because the power of a global test also depends on the direction of the alternative and the covariance structure of the entire process, the studentized tests are not guaranteed to outperform the original tests under every alternative. Section~\ref{sec:simu} examines the performance of both types of tests across different distributions in simulations.

\section{Covariate Adjustment}
\label{sec:covariate-adjustment}

In modern randomized experiments, researchers often have access to a rich set of pretreatment covariates, in addition to treatment assignment and outcomes. Although random assignment balances these characteristics on average, slight imbalances generally remain in realized samples.  When pretreatment covariates are predictive of outcomes, such information can be used to reduce extraneous sampling variation and improve the precision of treatment effect inference \citep{zhang2008improving,li2020rerandomization}.  This argument naturally extends to tests of treatment effect heterogeneity. Removing outcome variation explained by pretreatment covariates can make treatment effect variation more apparent \citep{ding2016randomization}.  In this section, we present covariate adjustment method for the centered expectile process that uses a fixed, finite-dimensional, prespecified covariate basis.

\begin{assumption}
\label{ass:covariate}
Let $Z=z(X)\in\mathbb R^p$, where $p$ is fixed and $z(\cdot)$ is prespecified.  Suppose $\E\|Z\|^2<\infty$ and $\Sigma_Z=\E[(Z-\E Z)(Z-\E Z)^\top]$ is positive definite.
\end{assumption}

A natural choice for $Z$ is $X$ itself, but it may also include indicators for categorical variables, transformations, or prespecified interaction terms. Although random assignment ensures covariate balance in the population, the sample means may differ between treatment and control groups. To characterize this sample imbalance, let $\bar Z_d=\E_{n,d}Z$ and define
\begin{equation}\label{eq:ca-imbalance}
B_n
=
\sqrt N(\bar Z_1-\bar Z_0).
\end{equation}

We consider using covariate imbalances in the sample to reduce sampling variation in a centered ETE process. Since random assignment guarantees that \(B_n\) has a mean of zero,  \(B_n=O_p(1)\). for any deterministic coefficient function \(a(\tau)\in\mathbb R^p\), the adjusted process $S_n(\tau)-a(\tau)^\top B_n$ has the same deterministic component $\sqrt N\,r(\tau)$ as the original process, but its asymptotic variance may differ.Therefore, we wish to find the coefficients $a(\tau)$ that minimize the pointwise asymptotic variance of the adjusted process.

The optimal coefficient can be obtained through linear projection. Write $\widetilde Z=Z-\E Z$, and let $\E_d$ denote expectation under the joint distribution of $(Y(d),Z)$ when covariates are involved. By the representation in \eqref{eq:S-linear}, the asymptotic covariance between $B_n$ and the sampling variation of the centered ETE process is
\[
C_\xi(\tau)
:=
\lim_{n\to\infty}
\Cov\!\left\{
B_n,\,
S_n(\tau)-\sqrt N\,r(\tau)
\right\}
=
(1-\lambda)\E_1\{\widetilde Z\,\xi_{1,\tau}\}
+
\lambda\E_0\{\widetilde Z\,\xi_{0,\tau}\}.
\]
Since $\Var(B_n)\to\Sigma_Z$, according to the least squares projection formula, the coefficient that minimizes the pointwise asymptotic variance among adjustments of the form $a(\tau)^\top B_n$ is $\gamma(\tau)=\Sigma_Z^{-1}C_\xi(\tau)$.  Therefore, subtracting \(\gamma(\tau)^\top B_n\) from the original process is equivalent to eliminating the sampling variation that can be linearly explained by the sample covariate imbalance. In practical applications, we estimate the above projection coefficients using sample correspondences. Let $\widetilde Z_{d,i}=Z_{d,i}-\bar Z_d$ and define
\begin{align}
\widehat\Sigma_Z
&=
(1-\lambda_n)\E_{n,1}
(\widetilde Z_1\widetilde Z_1^\top)
+
\lambda_n\E_{n,0}
(\widetilde Z_0\widetilde Z_0^\top),
\nonumber\\
\widehat C_\xi(\tau)
&=
(1-\lambda_n)\E_{n,1}
\{\widetilde Z_1\widehat\xi_{1,\tau}\}
+
\lambda_n\E_{n,0}
\{\widetilde Z_0\widehat\xi_{0,\tau}\},
\qquad
\widehat\gamma(\tau)
=
\widehat\Sigma_Z^{-1}\widehat C_\xi(\tau).
\label{eq:ca-feasible}
\end{align}
This yields the covariate-adjusted centered ETE process:
\begin{equation}\label{eq:ca-process}
S_n^{\mathrm{adj}}(\tau)
=
S_n(\tau)-\widehat\gamma(\tau)^\top B_n,
\qquad \tau\in\T,
\end{equation}
with corresponding feasible influence functions $\widehat\xi_{d,i,\tau}^{\mathrm{adj}} = \widehat\xi_{d,i,\tau} -\widehat\gamma(\tau)^\top\widetilde Z_{d,i}$.
Using \(S_n^{\mathrm{adj}}(\tau)\), the adjusted KS and CvM statistics can be constructed in the same way as the original test.

\begin{theorem}
\label{thm:covariate-adjustment}
Suppose Assumptions~\ref{ass:sampling}--\ref{ass:covariate} hold.  Then, uniformly over $\tau\in\T$,
\begin{equation}\label{eq:ca-linearization}
S_n^{\mathrm{adj}}(\tau)
=
\sqrt N\,r(\tau)
+
\sqrt{1-\lambda_n}\,
\mathbb G_{n,1}\xi_{1,\tau}^{\mathrm{adj}}
-
\sqrt{\lambda_n}\,
\mathbb G_{n,0}\xi_{0,\tau}^{\mathrm{adj}}
+
o_p(1),
\end{equation}
where $\xi_{d,\tau}^{\mathrm{adj}}(y,z) = \xi_{d,\tau}(y) - \gamma(\tau)^\top(z-\E Z)$. Consequently, under $H_0$,
\[
S_n^{\mathrm{adj}}
\rightsquigarrow
\mathbb G_{\mathrm{adj}}
\qquad\text{in }\ell^\infty(\T),
\]
where $\mathbb G_{\mathrm{adj}}$ is a centered Gaussian process with covariance kernel
\begin{equation}\label{eq:ca-kernel}
\Gamma_{\mathrm{adj}}(\tau,s)
=
\Gamma(\tau,s)
-
C_\xi(\tau)^\top\Sigma_Z^{-1}C_\xi(s),
\end{equation}
and $\Gamma$ is the covariance kernel in Theorem~\ref{thm:centered-limit}.  In particular, $\Gamma_{\mathrm{adj}}\preceq\Gamma$ in the positive semidefinite kernel order.
\end{theorem}

Theorem \(\ref{thm:covariate-adjustment}\) shows that covariate adjustment does not change the overall deviation function \(r\), and at the same time makes the asymptotic covariance kernel of the adjusted process no larger than the original process in the positive semi-definite sense. Therefore, this adjustment can reduce the asymptotic variance of any finite linear combination of processes. In particular, for a Wald-type test at a given expectile level, covariate adjustment increases the noncentrality parameter of the limiting chi-squared distribution and hence improves local asymptotic power (whenever the covariates are predictive of outcomes).  For the global KS and CvM statistics considered here, however, the comparison is more subtle. According to the Gaussian probability inequality of \citep{anderson1955integral}, when the covariance kernel decreases in the positive semi-definite sense, the probability of a centered Gaussian process falling into the same symmetric convex set does not decrease. This also applies to the corresponding acceptance fields of the KS and CvM statistics. However, the actual test needs to determine the critical value based on the null hypothesis distribution before and after adjustment, so this property does not guarantee that the adjusted test will always have higher power. As shown by \citet{neuhaus1976asymptotic} study of CvM-type tests, local power also depends on the relationship between the alternative directions and the asymptotic covariance structure. Nonetheless, covariate adjustment preserves the same deterministic signal while removing covariate-predictable sampling variation, which may suggest power gains in many empirically relevant settings.  We will examine the actual effect of this adjustment through simulation in the next section. Finally, the multiplier bootstrap and studentized refinement from Section~\ref{sec:inference} carry over directly to the adjusted process, and the parallel formulas are omitted for brevity.

\section{Monte Carlo Experiments}
\label{sec:simu}

In this section, we conduct Monte Carlo experiments to assess the finite-sample size and power of the proposed expectile multiplier tests relative to existing confidence set randomization, recentered subsampling, and martingale-transformed permutation methods. We further investigate the effects of pointwise studentization and covariate adjustment and compare computational costs across methods. Following the location-scale design of \citet{koenker2002inference}, which is also adopted in subsequent studies \citep{chernozhukov2005inference,chung2021permutation,chung2025quantile}, we generate potential outcomes as
\begin{equation}\label{eq:simulation-dgp}
Y_i(0)=\varepsilon_i,
\qquad
Y_i(1)=1+(1+\sigma_\delta)\varepsilon_i.
\end{equation}
The parameter $\sigma_\delta$ controls the degree of heterogeneity. We set $\sigma_\delta=0$ for null and $\sigma_\delta\in\{0.2,0.5\}$ for two levels of  heterogeneity. We draw the innovations $\varepsilon_i$ independently from four distributions: $N(0,1)$, $t_5$, $\mathrm{lognormal}(0,1)$, and the product normal $Z_{1i}Z_{2i}$, where $Z_{1i}$ and $Z_{2i}$ are independent standard normal variables. The $t_5$ and lognormal distributions allow us to examine heavy tails and skewness, respectively. In particular, the product normal distribution has a continuous CDF but a density that diverges to infinity as $x\to0$. This makes it a useful stress test for our procedure, which requires only finite second moments and CDF continuity near the relevant expectiles, without the density regularity conditions often imposed by existing methods.

To examine distribution differences beyond the scale changes in \eqref{eq:simulation-dgp}, we further consider alternatives between the normal distribution and the $t$ distribution with 5 degrees of freedom, including both variance matching and variance mismatching settings. Let $U_i\sim U(0,1)$ and set
\begin{equation}\label{eq:simulation-shape-dgp}
Y_i(0)=\Phi^{-1}(U_i),
\qquad
Y_i(1)=1+aF_{t_5}^{-1}(U_i),
\qquad
a\in\left\{\sqrt{\frac{3}{5}},1\right\}.
\end{equation}
When $a=\sqrt{3/5}$, the marginal distributions of the treatment group and the control group have the same variance, so that the impact of the difference in distribution shape can be examined separately; when $a=1$, the two groups of distributions are different in both shape and scale. 

Throughout the Monte Carlo experiments, we consider balanced treatment allocation ($n_1/n=1/2$) with total sample sizes $n\in\{100,200,400,800\}$. We also examine unbalanced allocation ($n_1/n=1/3$) with $n\in\{150,300,600\}$, for which the results are reported in Section \ref{app:additional-monte-carlo} of the Online Appendix. Each design uses 1,000 Monte Carlo replications, and all tests are conducted at the $5\%$ nominal level. We compare the proposed expectile tests with existing procedures using the following implementations:

\begin{itemize}[leftmargin=1.95cm,labelsep=0.4em,itemsep=0.25em,
  topsep=0.4em,parsep=0pt]

\item[\textbf{eMBT}]
The proposed expectile multiplier bootstrap tests, including KS and CvM statistics and their pointwise studentized counterparts (S--KS and S--CvM). All four variants use 81 equally spaced expectile indices over $[0.10,0.90]$ and 999 Mammen multiplier draws.

\item[\textbf{FRT-CI}]
The Fisher randomization test with a confidence interval proposed by \citet{ding2016randomization}. The procedure maximizes the randomization $p$-value over a $(1-\gamma)$ confidence interval for the unknown treatment effect and adds $\gamma$ to obtain the adjusted $p$-value. We set $\gamma=0.001$ and use 999 permutations.

\item[\textbf{RS}]
The subsampling test of \citet{chernozhukov2005inference}, based on the recentered quantile treatment effect process. We use 250 subsamples of size $b=\lfloor 20+n^{1/4}\rfloor$ and the quantile grid $\{0.05,0.10,\ldots,0.95\}$, with pointwise standard error weighting and the authors' finite-sample critical value adjustment.

\item[\textbf{D-mtPT}]
The CDF-based martingale-transformed permutation test of \citet{chung2021permutation}, implemented using \texttt{PT.Khmaladze.fit} with the default \texttt{akj.control} settings and 999 permutations.

\item[\textbf{Q-mtPT}]
The quantile-based martingale-transformed permutation test of \citet{chung2025quantile}, implemented using \texttt{PTQTE.Khmaladze.fit} with 999 permutations over the quantile grid $\{0.10,0.15,\ldots,0.90\}$.\footnote{Both procedures are implemented using \texttt{RATest} version 0.1.12.}
\end{itemize}

Table~\ref{tab:balanced-size} reports the empirical rejection frequencies of each test under the null hypothesis. The four eMBT tests show stable size control across all four distributions, including the product normal distribution with an unbounded density. Although some overrejection occurs at $n=100$, especially under lognormality, the empirical rejection frequencies of all eMBT tests lie between 0.038 and 0.065 for $n\ge 200$, with similar performance for the original and studentized tests. In contrast, FRT-CI tends to be conservative, consistent with its confidence set construction, while RS exhibits greater variation across distributions. D-mtPT and Q-mtPT are reasonably well sized under normal and $t_5$ distributions, but show some overrejection under product normality and more pronounced overrejection under lognormality, especially in smaller samples.

Table ~\ref{tab:balanced-power} reports the empirical rejection frequencies under the alternatives. Overall, the eMBT tests perform well compared to existing procedures, often achieving the highest empirical power. For the scale alternatives in Panels A and B, the proposed tests generally outperform the comparison procedures under normal and $t_5$ distributions, especially at moderate sample sizes. Under lognormality, the eMBT remains competitive, especially after studentization, although D-mtPT and Q-mtPT often exhibit higher rejection frequencies. These comparisons should be interpreted in light of some size distortions under the null hypothesis (see Table ~\ref{tab:balanced-size}). The advantage of the eMBT is particularly pronounced for product normal innovations. FRT-CI exhibits little power against the scale alternatives and D-mtPT has relatively low rejection frequencies. In contrast, the proposed tests consistently perform best. Panels C and D extend the comparison to normal-to-$t_5$ alternatives, both with and without variance matching. Q-mtPT achieves the highest empirical power in the variance-matched design. Conversely, eMBT is the most powerful across all sample sizes when the alternative involves a scale difference. These results suggest that expectile-based tests are sensitive to both scale and shape departures,  with particularly strong power against scale changes, while quantile-based procedures may be more sensitive to certain variance-preserving shape alternatives.

\begin{table}[!t]
\centering
\begin{threeparttable}
\caption{Empirical size under balanced assignment}
\label{tab:balanced-size}
\small
\setlength{\tabcolsep}{3.5pt}
\renewcommand{\arraystretch}{1.04}
\begin{tabular*}{\textwidth}{@{\extracolsep{\fill}}llrrrrrrrr@{}}
\toprule
&&\multicolumn{4}{c}{eMBT}&
\multicolumn{4}{c}{Comparison tests}\\
\cmidrule(lr){3-6}\cmidrule(lr){7-10}
Innovation & $n$ & KS & CvM & S--KS & S--CvM &
FRT-CI & RS & D-mtPT & Q-mtPT\\
\midrule
\multirow{4}{*}{Normal}
&100&0.069&0.068&0.068&0.063&0.015&0.056&0.055&0.042\\
&200&0.048&0.051&0.053&0.051&0.035&0.075&0.039&0.044\\
&400&0.045&0.052&0.052&0.052&0.030&0.088&0.051&0.049\\
&800&0.038&0.043&0.044&0.044&0.040&0.074&0.045&0.041\\
\addlinespace[2pt]
\multirow{4}{*}{$t_5$}
&100&0.048&0.053&0.063&0.056&0.015&0.045&0.043&0.045\\
&200&0.043&0.057&0.058&0.061&0.029&0.041&0.036&0.055\\
&400&0.047&0.055&0.051&0.056&0.026&0.047&0.048&0.039\\
&800&0.045&0.052&0.042&0.050&0.025&0.031&0.050&0.045\\
\addlinespace[2pt]
\multirow{4}{*}{Lognormal}
&100&0.088&0.066&0.088&0.064&0.034&0.027&0.191&0.200\\
&200&0.054&0.053&0.057&0.045&0.038&0.023&0.161&0.136\\
&400&0.061&0.048&0.051&0.047&0.041&0.029&0.122&0.107\\
&800&0.047&0.047&0.050&0.049&0.044&0.039&0.094&0.063\\
\addlinespace[2pt]
\multirow{4}{*}{Product normal}
&100&0.055&0.052&0.073&0.048&0.032&0.046&0.075&0.064\\
&200&0.059&0.053&0.065&0.053&0.030&0.036&0.073&0.073\\
&400&0.059&0.054&0.057&0.057&0.049&0.039&0.072&0.069\\
&800&0.046&0.048&0.053&0.051&0.048&0.031&0.063&0.073\\
\bottomrule
\end{tabular*}
\end{threeparttable}
\end{table}

\begin{table}[htbp]
\centering
\begin{threeparttable}
\caption{Empirical power under balanced assignment}
\label{tab:balanced-power}
\small
\setlength{\tabcolsep}{3.5pt}
\renewcommand{\arraystretch}{0.9}
\begin{tabular*}{\textwidth}{@{\extracolsep{\fill}}llrrrrrrrr@{}}
\toprule
&&\multicolumn{4}{c}{eMBT}&
\multicolumn{4}{c}{Comparison tests}\\
\cmidrule(lr){3-6}\cmidrule(lr){7-10}
Innovation & $n$ & KS & CvM & S--KS & S--CvM &
FRT-CI & RS & D-mtPT & Q-mtPT\\
\midrule
\multicolumn{10}{l}{\textit{Panel A: $\sigma_\delta=0.2$}}\\[2pt]
\multirow{4}{*}{Normal}
&100&0.220&0.239&0.235&0.232&0.049&0.144&0.082&0.125\\
&200&0.359&0.401&0.372&0.393&0.112&0.295&0.159&0.214\\
&400&0.649&0.692&0.662&0.678&0.226&0.511&0.334&0.400\\
&800&0.932&0.957&0.944&0.953&0.535&0.803&0.627&0.703\\
\addlinespace[2pt]
\multirow{4}{*}{$t_5$}
&100&0.130&0.144&0.158&0.158&0.028&0.084&0.085&0.094\\
&200&0.201&0.261&0.264&0.284&0.070&0.147&0.133&0.195\\
&400&0.364&0.462&0.467&0.494&0.133&0.231&0.217&0.358\\
&800&0.616&0.751&0.756&0.794&0.317&0.465&0.445&0.603\\
\addlinespace[2pt]
\multirow{4}{*}{Lognormal}
&100&0.099&0.083&0.103&0.085&0.054&0.050&0.294&0.270\\
&200&0.109&0.113&0.141&0.136&0.118&0.067&0.387&0.349\\
&400&0.128&0.136&0.178&0.161&0.180&0.092&0.452&0.360\\
&800&0.144&0.195&0.281&0.253&0.326&0.185&0.551&0.494\\
\addlinespace[2pt]
\multirow{4}{*}{Product normal}
&100&0.082&0.096&0.126&0.099&0.040&0.063&0.082&0.091\\
&200&0.134&0.175&0.174&0.179&0.032&0.079&0.070&0.120\\
&400&0.224&0.272&0.275&0.297&0.032&0.120&0.061&0.198\\
&800&0.402&0.510&0.501&0.547&0.038&0.240&0.097&0.311\\
\midrule
\multicolumn{10}{l}{\textit{Panel B: $\sigma_\delta=0.5$}}\\[2pt]
\multirow{4}{*}{Normal}
&100&0.725&0.791&0.756&0.787&0.219&0.535&0.202&0.470\\
&200&0.959&0.983&0.968&0.977&0.577&0.864&0.495&0.801\\
&400&1.000&1.000&0.999&1.000&0.935&0.996&0.815&0.980\\
&800&1.000&1.000&1.000&1.000&1.000&1.000&0.995&1.000\\
\addlinespace[2pt]
\multirow{4}{*}{$t_5$}
&100&0.425&0.544&0.575&0.587&0.128&0.315&0.202&0.375\\
&200&0.722&0.843&0.856&0.879&0.350&0.552&0.410&0.688\\
&400&0.947&0.988&0.989&0.993&0.739&0.882&0.734&0.953\\
&800&0.998&1.000&1.000&1.000&0.985&0.998&0.975&1.000\\
\addlinespace[2pt]
\multirow{4}{*}{Lognormal}
&100&0.178&0.195&0.244&0.230&0.173&0.144&0.560&0.568\\
&200&0.253&0.307&0.407&0.373&0.390&0.246&0.698&0.712\\
&400&0.348&0.448&0.654&0.568&0.681&0.500&0.881&0.899\\
&800&0.575&0.720&0.923&0.856&0.951&0.880&0.991&0.990\\
\addlinespace[2pt]
\multirow{4}{*}{Product normal}
&100&0.285&0.334&0.386&0.362&0.032&0.197&0.127&0.244\\
&200&0.501&0.618&0.628&0.651&0.037&0.283&0.122&0.425\\
&400&0.783&0.880&0.875&0.904&0.059&0.529&0.158&0.636\\
&800&0.972&0.991&0.992&0.996&0.057&0.851&0.256&0.904\\
\midrule
\multicolumn{10}{l}{\textit{Panel C: variance-matched normal-to-$t_5$ alternative}}\\[2pt]
\multirow{4}{*}{Normal/$t_5$}
&100&0.075&0.091&0.102&0.096&0.038&0.077&0.020&0.107\\
&200&0.072&0.092&0.109&0.108&0.073&0.138&0.030&0.139\\
&400&0.082&0.128&0.143&0.155&0.183&0.231&0.064&0.290\\
&800&0.089&0.174&0.239&0.250&0.370&0.408&0.127&0.490\\
\midrule
\multicolumn{10}{l}{\textit{Panel D: unstandardized normal-to-$t_5$ alternative}}\\[2pt]
\multirow{4}{*}{Normal/$t_5$}
&100&0.166&0.168&0.174&0.157&0.027&0.094&0.056&0.079\\
&200&0.313&0.333&0.307&0.314&0.039&0.152&0.087&0.095\\
&400&0.609&0.636&0.581&0.568&0.065&0.222&0.128&0.150\\
&800&0.904&0.897&0.877&0.855&0.096&0.374&0.257&0.271\\
\bottomrule
\end{tabular*}
\end{threeparttable}
\end{table}

Table~\ref{tab:balanced-power} also shows how pointwise studentization affects power. While the impact is modest under normal distribution, studentization often yields substantial gains under $t_5$, product normal, and lognormal distributions, as well as under the variance-matched normal-to-$t_5$ alternative. Although some power losses occur, particularly in Panel D, they are generally limited compared to the gains observed elsewhere. Overall, the results favor studentization. Figure~\ref{fig:studentization-diagnostic} plots the pointwise standard deviation $v_A(\tau)$ of the regular process $A_n(\tau)$ and the standardized scale signal $s(\tau)=m(\tau)/v_A(\tau)$, where $m(\tau)=\E|\varepsilon-e_\varepsilon(\tau)|$, each normalized by its value at $1/2$. Under normality, both profiles are relatively flat. Under $t_5$ and product normal innovations, pointwise variability increases toward the tails, and the standardized signal weakens. Under lognormality, the profiles exhibit pronounced asymmetry: the pointwise standard deviation rises sharply toward the upper expectiles, while the standardized scale signal declines over the same range. These patterns suggest that coordinates with greater sampling variability do not necessarily carry correspondingly stronger standardized signals. Thus, by placing expectile coordinates on a common noise scale, studentization reduces the relative influence of noisy regions and enhances the power of global tests against scale departures.

\begin{figure}[t]
\centering
\includegraphics[width=0.9\textwidth]{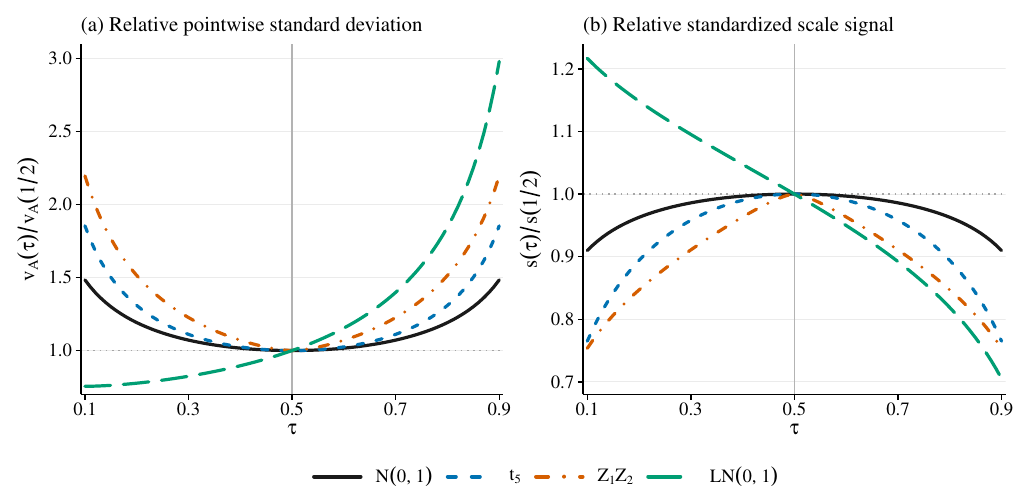}
\caption{Population diagnostics for pointwise studentization.
Panel (a) shows the pointwise standard deviation $v_A(\tau)$
of the regular expectile process, and Panel (b) shows the
standardized scale signal $s(\tau)=m(\tau)/v_A(\tau)$.
Both profiles are normalized by their values at $\tau=1/2$.}
\label{fig:studentization-diagnostic}
\end{figure}

We next examine the finite-sample benefits of covariate adjustment using a design with a binary pretreatment covariate. Specifically, let $X_i\sim\operatorname{Bernoulli}(1/2)$ and $\varepsilon_i\sim N(0,1)$ be independent, and set
\begin{equation}\label{eq:simulation-covariate-dgp}
\sigma(X_i)=
\begin{cases}
1/\sqrt{5},&X_i=0,\\
3/\sqrt{5},&X_i=1,
\end{cases}
\qquad
Y_i(0)=\sigma(X_i)\varepsilon_i,
\qquad
Y_i(1)=1+(1+\gamma)\sigma(X_i)\varepsilon_i.
\end{equation}
The conditional scale is three times larger for $X_i=1$ than for $X_i=0$, while $\E\{\sigma^2(X_i)\}=1$. When $\gamma=0$, the treatment effect is constant at one, so the null holds. Positive values of $\gamma$ introduce a treatment-induced scale change while preserving the average treatment effect of one. We use $z(X)=X$ as the adjustment basis, and compare each eMBT statistic with its covariate-adjusted counterpart under the same simulation settings.

\begin{table}[htbp]
\centering
\begin{threeparttable}
\caption{Comparison of eMBT tests with and without covariate adjustment}
\label{tab:covariate-adjustment}
\small
\renewcommand{\arraystretch}{1.0}
\begin{tabular*}{\textwidth}{@{\extracolsep{\fill}}lrrrrrrrr@{}}
\toprule
&\multicolumn{4}{c}{Unadjusted eMBT}&
 \multicolumn{4}{c}{Covariate-adjusted eMBT}\\
\cmidrule(lr){2-5}\cmidrule(lr){6-9}
$n$ & KS & CvM & S--KS & S--CvM & KS & CvM & S--KS & S--CvM\\
\midrule
\multicolumn{9}{l}{\textit{Panel A: $\gamma=0$ (size)}}\\[2pt]
100&0.071&0.068&0.091&0.068&0.066&0.068&0.086&0.072\\
200&0.050&0.051&0.068&0.049&0.053&0.052&0.065&0.054\\
400&0.056&0.053&0.056&0.048&0.053&0.057&0.060&0.061\\
800&0.047&0.045&0.043&0.042&0.057&0.058&0.050&0.055\\
\midrule
\multicolumn{9}{l}{\textit{Panel B: $\gamma=0.2$}}\\[2pt]
100&0.153&0.163&0.180&0.164&0.173&0.194&0.214&0.202\\
200&0.189&0.231&0.227&0.243&0.219&0.293&0.293&0.315\\
400&0.337&0.405&0.397&0.429&0.380&0.488&0.487&0.534\\
800&0.603&0.687&0.671&0.711&0.689&0.784&0.799&0.834\\
\midrule
\multicolumn{9}{l}{\textit{Panel C: $\gamma=0.5$}}\\[2pt]
100&0.398&0.488&0.496&0.509&0.447&0.548&0.563&0.613\\
200&0.679&0.778&0.771&0.808&0.760&0.875&0.878&0.907\\
400&0.945&0.972&0.974&0.986&0.966&0.995&0.994&0.996\\
800&0.999&1.000&1.000&1.000&1.000&1.000&1.000&1.000\\
\bottomrule
\end{tabular*}
\begin{tablenotes}[flushleft]
\footnotesize
\item Notes: Entries are rejection frequencies at the $5\%$ nominal level
based on 1,000 Monte Carlo replications of
\eqref{eq:simulation-covariate-dgp}, with $n_1=n_0=n/2$.   The
adjustment basis is $z(X)=X$.  S denotes pointwise studentization.  
\end{tablenotes}
\end{threeparttable}
\end{table}

The results are summarized in Table~\ref{tab:covariate-adjustment}.Covariate adjustment has little effect on the empirical size but clearly improves test power in both alternatives. The power of the four test statistics improved regardless of whether pointwise studentization was performed. For example, when $n=800$, $\gamma=0.2$, the empirical power of the CvM test increases from 0.687 to 0.784, while the power of the S–CvM test increases from 0.711 to 0.834. Notably, in this design, pretreatment covariates affect the conditional scale of the potential outcome but not its conditional mean. Thus, the increase in power comes not from the covariate's predictive power on the outcome mean, but from its predictive effect on the expectile influence functions. This is consistent with the asymptotic variance reduction results established by the Theorem~\ref{thm:covariate-adjustment}.

\begin{table}[htbp]
\centering
\begin{threeparttable}
\caption{Mean single-core computation time (seconds)}
\label{tab:computation-time}
\small
\setlength{\tabcolsep}{9pt}
\renewcommand{\arraystretch}{1.05}
\begin{tabular}{lrrrr}
\toprule
Method & $n=100$ & $n=200$ & $n=400$ & $n=800$ \\
\midrule
eMBT            & 0.010 & 0.017 & 0.031 & 0.060 \\
FRT-CI          & 12.178 & 13.873 & 17.098 & 24.206 \\
RS              & 0.229 & 0.232 & 0.228 & 0.235 \\
D-mtPT          & 0.295 & 0.377 & 0.620 & 1.356 \\
Q-mtPT          & 5.008 & 6.740 & 10.263 & 19.164 \\
\bottomrule
\end{tabular}
\begin{tablenotes}[flushleft]
\footnotesize
\item Notes: Entries are mean elapsed seconds over 100 independently generated
datasets for each balanced normal-null design.  The eMBT row reports a joint
call computing all four reported eMBT variants.  Computations use one core of
an Apple M4 Pro chip with 24 GB of memory.
\end{tablenotes}
\end{threeparttable}
\end{table}

Finally, we compare the computational costs of the proposed tests with those of the competing procedures. The results are reported in Table~\ref{tab:computation-time}. Computing all four eMBT variants jointly takes only 0.010 seconds at $n=100$ and 0.060 seconds at $n=800$, which reflects the computational simplicity of our procedure: it requires no nonparametric nuisance estimation, and the multiplier bootstrap reweights precomputed influence function quantities without repeated re-estimation.  In contrast, the competing procedures involve more computationally intensive steps, including subsampling, repeated permutation calculations, and nonparametric nuisance estimation. At $n=800$, their computation times range from 0.235 seconds for RS to 24.206 seconds for FRT-CI. Overall, the simulations demonstrate that eMBT reliably controls size and achieves competitive power across distributions with substantially different tail behavior, skewness, and density characteristics. Covariate adjustment can further enhance power by exploiting baseline information without altering the population discrepancy being tested. These properties, together with eMBT's tuning-free construction and computationally efficient multiplier bootstrap, make it a practical and flexible tool for empirical studies of treatment effect heterogeneity.

\FloatBarrier

% =====================================================================
\section{Empirical Application}
\label{sec:empirical}

In this section, we revisit the large-scale randomized field experiment of \citet{ferraro2013using},\footnote{The data and original replication files are publicly available through Harvard Dataverse at \url{https://doi.org/10.7910/DVN1/22633}.} conducted in partnership with the Cobb County Water System in metropolitan Atlanta to examine how nonpecuniary incentives affect residential water consumption. More than 100,000 households were randomly assigned within meter reading routes to a control group or one of three interventions. The first, \textit{Technical Advice}, provided practical water saving tips. The second, \textit{Weak Social Norm}, supplemented these tips with an appeal to households' responsibility to conserve water. The third, \textit{Strong Social Norm}, further provided personalized information on each household's previous water consumption relative to the county median and its percentile rank among county residents.

The average effects of these interventions and heterogeneity in household responses were subsequently studied further. \citet{ferraro2013using} found that simply providing water-saving advice had little effect, while prosocial appeals and personalized social comparisons produced progressively larger average water savings. Personalized social comparison was particularly effective among households with higher baseline water consumption. \citet{ferraro2013heterogeneous} further examined treatment effects across household subgroups and quantiles. They found pronounced heterogeneity under the strong social norm intervention, while evidence of heterogeneity under the other two information interventions was relatively limited. More recently, \citet{oka2026regression} found that the strong social norm intervention shifted probability mass from higher to more moderate levels of water consumption. However, existing research has not formally tested whether each intervention merely shifts the marginal distribution of household water use by an unknown common location shift or also changes the shape of the distribution.

To investigate this question, we apply the proposed eMBT tests to each intervention relative to the control group. Following \citet{ferraro2013using}, the outcome is total household water consumption during the summer of 2007, measured in thousands of gallons and constructed from July--October billing records.\footnote{We follow the sample exclusions in \citet{ferraro2013using} for discontinued water service, unusable billing records, and catastrophic leaks. The resulting sample reproduces the group sizes and mean outcomes reported in their Table~1. Due to billing lags, July--October records correspond to June--September consumption.} The resulting sample contains 71,643 control households, 11,675 households assigned to Technical Advice, 11,675 assigned to Weak Social Norm, and 11,676 assigned to Strong Social Norm. For each comparison, we implement the KS and CvM tests and their pointwise studentized versions, both with and without covariate adjustment. For adjustment, we use the two pretreatment water consumption measures included in Model A of \citet{ferraro2013using}, covering the previous watering season and spring 2007. All $p$-values are computed using 9,999 Mammen multiplier bootstrap draws.

\begin{table}[t]
\centering
\begin{threeparttable}
\caption{Tests for treatment effect heterogeneity in household water consumption}
\label{tab:water-application}
\small
\setlength{\tabcolsep}{4pt}
\begin{tabular*}{\textwidth}{@{\extracolsep{\fill}}lcccccccc@{}}
\toprule
& \multicolumn{4}{c}{Unadjusted} & \multicolumn{4}{c}{Covariate adjusted} \\
\cmidrule(lr){2-5}\cmidrule(lr){6-9}
Intervention & KS & CvM & S--KS & S--CvM
             & KS & CvM & S--KS & S--CvM \\
\midrule
Technical Advice
  & 0.3442 & 0.4216 & 0.4882 & 0.6040
  & 0.4147 & 0.5060 & 0.6174 & 0.7705 \\
Weak Social Norm
  & 0.4853 & 0.2361 & 0.0630 & 0.1084
  & 0.4986 & 0.1906 & \textbf{0.0131} & \textbf{0.0453} \\
Strong Social Norm
  & \textbf{0.0002} & \textbf{0.0001} & \textbf{0.0001} & \textbf{0.0001}
  & \textbf{0.0001} & \textbf{0.0001} & \textbf{0.0001} & \textbf{0.0001} \\
\bottomrule
\end{tabular*}
\begin{tablenotes}[flushleft]
\footnotesize
\item Notes:  Entries are $p$-values based on 9,999 Mammen multiplier bootstrap draws. 
Bold entries indicate $p\leq 0.05$. 
\end{tablenotes}
\end{threeparttable}
\end{table}

Table~\ref{tab:water-application} reports the results of tests for treatment effect heterogeneity under each of the three interventions. For Technical Advice, none of the tests indicates significant heterogeneity. By contrast, all eight tests reject for Strong Social Norm, with $p$-values no greater than 0.0002, providing strong evidence of heterogeneous responses to personalized social comparisons. For Weak Social Norm, the results are less uniform. The unadjusted KS test yields a $p$-value of 0.4853; pointwise studentization reduces it to 0.0630, and incorporating covariate adjustment further reduces it to 0.0131. The covariate-adjusted S--CvM test also rejects the null hypothesis at the 5\% significance level, whereas neither unstudentized test rejects, even after covariate adjustment. We view these findings as providing suggestive evidence of treatment effect heterogeneity under Weak Social Norm, although less robust than that for Strong Social Norm. As shown by the simulation results in Section~\ref{sec:simu}, tests incorporating studentization or covariate adjustment often attain higher empirical power, which may help explain the different rejection outcomes for Weak Social Norm. In our empirical analysis, covariate adjustment yields median pointwise variance reductions of approximately 43--47\% across the three comparisons. Similar precision gains from regression adjustment are reported by \citet{oka2026regression} in their analysis of the same experiment.

To further understand these results, Figure~\ref{fig:water-application} plots the mean-anchored ETE curves for the three interventions with and without covariate adjustment. The unadjusted curve for Strong Social Norm slopes sharply downward, from about 1.0 thousand gallons at $\tau=0.1$ to $-2.1$ thousand gallons at $\tau=0.9$. Since this curve measures the treatment effect at each expectile relative to the average treatment effect, its downward slope indicates smaller water savings at lower expectiles and larger savings at higher expectiles. In other words, after aligning the treatment and control distributions at their means, the lower expectiles move relatively upward and the upper expectiles move relatively downward, with a larger relative movement at the upper end. As a result, the $0.1$--$0.9$ inter-expectile spread is about 3.1 thousand gallons narrower under treatment. Weak Social Norm exhibits a similar pattern, but to a lesser extent, while Technical Advice shows no comparable contraction. The overall shapes of these curves remain essentially unchanged after covariate adjustment.

These patterns are consistent with \citet[Table~3]{ferraro2013using}, who report average water savings of 2.35 thousand gallons among households with higher baseline consumption, compared with only 0.57 thousand gallons among households with lower baseline consumption. Our results also complement the distributional changes in water consumption documented by \citet{oka2026regression}. One possible economic interpretation, consistent with the evidence on treatment mechanisms in \citet{ferraro2013heterogeneous}, is that personalized social comparisons make conservation norms more salient to high-use households. These households may also have greater scope to reduce discretionary consumption, particularly outdoor water use. The differences across the three interventions suggest that personalized social comparisons reduce average water consumption while also narrowing the gap between higher and lower consumption levels.

\begin{figure}[t]
\centering
\includegraphics[width=0.9\textwidth]{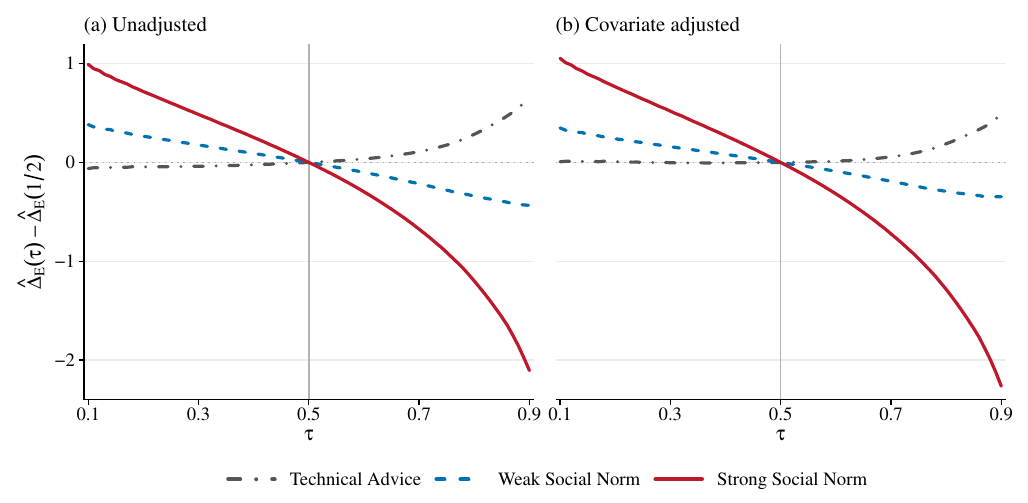}
\caption{Mean-anchored expectile treatment effect curves. Each curve plots $\widehat\Delta(\tau)-\widehat\Delta(1/2)$, measured in thousands of gallons, for an intervention relative
to the control group. The left and right panels show unadjusted and covariate-adjusted estimates, respectively.}
\label{fig:water-application}
\end{figure}

\FloatBarrier

\section{Conclusion}
\label{sec:con}

This paper develops tuning-free tests for treatment effect heterogeneity via the sample expectile process. By centering the difference between the expectile curves of the treatment and control groups around its mean, we test whether the treatment induces changes in the marginal outcome distribution that go beyond a common location shift. We establish the asymptotic properties of the resulting process, accounting for estimation of the unknown common effect, and construct a fast multiplier bootstrap that requires no density or score estimation, bandwidth selection, or repeated refitting. Furthermore, we introduce a pointwise studentized test and a covariate adjustment method.  Monte Carlo simulations demonstrate that the proposed tests generally achieve reliable size control, competitive power, and substantial computational savings compared with existing methods.
 In a large-scale field experiment on residential water conservation, our results indicate that personalized social comparison information triggers substantial heterogeneity in treatment effects. For a weaker social-norm intervention, additional evidence of heterogeneity emerges after applying studentization and covariate adjustment.

Several avenues remain for future research.  Extending the proposed tests to observational studies would broaden their applicability beyond randomized experiments, while adapting the framework to panel data and difference-in-differences settings would allow for the analysis of distributional treatment effects over time. Additionally, one could develop tests for stochastic dominance and other partial orderings of marginal distributions of potential outcomes based on expectiles. This would allow researchers to go beyond merely identifying treatment effect heterogeneity and instead determine the direction of shifts in the distribution. We leave these interesting extensions for future research.

\begingroup \setstretch{1}
\setlength{\bibsep}{0pt}
\putbib
\endgroup
\end{bibunit}

\clearpage
\begin{bibunit}
\appendix
\setstretch{1.2}
\renewcommand{\thesection}{\Alph{section}}
\renewcommand{\thesubsection}{\thesection.\arabic{subsection}}
\counterwithin{equation}{section}
\counterwithin{figure}{section}
\counterwithin{table}{section}

\setcounter{page}{1}
\pagenumbering{arabic}

\begin{center}
{\LARGE Online Appendix for}\\[1ex]
{\Large ``Tuning-Free Tests for Treatment Effect Heterogeneity\\
via Expectile Processes''}\\[3ex]
Haokun Lu \quad Xiaojun Song\\[1ex]
\end{center}

\vspace{1em}

This Online Appendix is organized as follows. Section \ref{app:proofs} provides the proofs of the theoretical results in the main text, while Section \ref{app:additional-monte-carlo} presents additional Monte Carlo results.

\section{Proofs}
\label{app:proofs}
% =====================================================================
\subsection{Proofs for Section~\ref{sec:framework}}
\label{app:framework}

\paragraph{Proof of Proposition~\ref{prop:full-curve}}
\begin{proof}
First, assume that $H_0$ holds; that is, for some $\delta\in\R$, $Y(1)\stackrel{d}{=}Y(0)+\delta$. By the translation equivariance of expectiles, $e_1(\tau)=e_0(\tau)+\delta$ for every $\tau\in(0,1)$; thus, the forward implication holds.

Next, we prove the converse implication. To do so, it suffices to show that a distribution with a finite first moment is uniquely determined by its full expectile curve. When $\E|Y|<\infty$, $e(\tau)$ is the unique solution to
\[
\E\psi_\tau\{Y,e(\tau)\}=0,
\qquad
\psi_\tau(y,e)=(y-e)+(2\tau-1)|y-e|.
\]
Let $\mu=\E Y=e(1/2)$, and define $M_F(x)=\E_F|Y-x|$. Evaluating the score equation at $e(\tau)$ yields
\begin{equation*}
0=\E\psi_\tau\{Y,e(\tau)\}=\mu-e(\tau)+(2\tau-1)M_F\{e(\tau)\},
\end{equation*}
which implies
\begin{equation}\label{eq:expectile-mad-identity}
e(\tau)-\mu
=
(2\tau-1)M_F\{e(\tau)\}.
\end{equation}
Therefore, the expectile curve determines $M_F$ at every point in its range.

Next, we characterize the range of the expectile curve. Assume that $F$ is nondegenerate, and let $M(\tau,e)=\E\psi_\tau(Y,e)$. For any given $(y,\tau)$, the mapping $e\mapsto\psi_\tau(y,e)$ is continuous and piecewise linear, with slopes $-2\tau$ and $-2(1-\tau)$ on its two linear pieces, both strictly negative. Consequently, for $e_2>e_1$,
\begin{equation}
\label{eq:app-strong-monotonicity-original}
-2\max\{\tau,1-\tau\}(e_2-e_1)
\le
\psi_\tau(y,e_2)-\psi_\tau(y,e_1)
\le
-2\min\{\tau,1-\tau\}(e_2-e_1)<0.
\end{equation}
Taking expectations in this inequality shows that the mapping $e\mapsto M(\tau,e)$ is also continuous and strictly decreasing. If $\tau_2>\tau_1$, then
\begin{align*}
M\{\tau_2,e(\tau_1)\}
&=
M\{\tau_1,e(\tau_1)\}
+
2(\tau_2-\tau_1)M_F\{e(\tau_1)\} \\
&=
2(\tau_2-\tau_1)M_F\{e(\tau_1)\}
>
0
=
M\{\tau_2,e(\tau_2)\}.
\end{align*}
Since $M(\tau_2,\cdot)$ is strictly decreasing and equals zero at $e(\tau_2)$, it follows that
\[
e(\tau_2)>e(\tau_1).
\]
Furthermore, the joint continuity of $M(\tau,e)$ and the uniqueness of its zero imply that the mapping $\tau\mapsto e(\tau)$ is continuous.

Let $a=\operatorname*{ess\,inf} Y$ and $b=\operatorname*{ess\,sup}Y$, where either endpoint may be infinite. The score equation can also be written as
\[
\tau\,\E\{Y-e(\tau)\}_+
=
(1-\tau)\,\E\{e(\tau)-Y\}_+.
\]
Monotonicity ensures that the expectile curve has limits at the endpoints of $(0,1)$, and the balance relation above rules out the possibility that these limits lie strictly inside the support of the distribution. Consequently, as $\tau\downarrow0$, $e(\tau)\downarrow a$, and as $\tau\uparrow1$, $e(\tau)\uparrow b$. Thus, the range of the continuous and strictly increasing mapping $\tau\mapsto e(\tau)$ is $(a,b)$. For every $x\in(a,b)$, there exists a unique $\tau(x)$ such that $e\{\tau(x)\}=x$. By \eqref{eq:expectile-mad-identity},
\begin{equation}\label{eq:expectile-recovers-mad}
M_F(x)
=
\frac{x-\mu}{2\tau(x)-1},
\qquad x\ne\mu.
\end{equation}
At $x=\mu$ and at any finite endpoint, the value of $M_F$ is determined by continuity because $|M_F(x)-M_F(z)|\le |x-z|$.

Finally, it remains to recover $F$ from $M_F$. The function $M_F$ is convex, and its one-sided derivatives satisfy
\[
M'_{F,+}(x)=2F(x)-1,
\qquad
M'_{F,-}(x)=2F(x-)-1.
\]
Therefore, $M_F$ determines $F(x)$ and $F(x-)$ at every $x$, including the probability mass of any atom, thereby uniquely determining the entire distribution. For a degenerate distribution, this conclusion can be verified separately. If the expectile curve is identically equal to a constant $c$, then $\mu=c$. Evaluating \eqref{eq:expectile-mad-identity} at any $\tau\ne1/2$ yields $M_F(c)=0$; thus, $Y=c$ almost surely, and the identification result remains valid in this case.

Now assume that $e_1(\tau)-e_0(\tau)=\delta$ for every $\tau\in(0,1)$. By translation equivariance, the expectile curve of $Y(0)+\delta$ is $e_0(\tau)+\delta=e_1(\tau)$. By the identification result established above,
\[
Y(1)\stackrel{d}{=}Y(0)+\delta,
\]
which is equivalent to $H_0$.
\end{proof}

% =====================================================================
\subsection{Proofs for Section~\ref{sec:asymptotic}}
\label{app:asymptotic}
% =====================================================================

We first record two useful lemmas used throughout the proofs.  For a single distribution we suppress the arm subscript, let $P$ denote the law of $Y$, and write $\E_n$ for the empirical average.  For any probability law $Q$, let $\|f\|_{Q,2}:=(\int f^2\,dQ)^{1/2}$; for an empirical average, write $\|f\|_{\E_n,2}:=(\E_n f^2)^{1/2}$.  Let $c_{\T}=\min\{\tau_L,1-\tau_U\}>0$.  Constants below may change from line to line but do not depend on $n$ or $\tau\in\T$.

\begin{lemma}\label{lem:appendix-regularity}
Under Assumption~\ref{ass:regularity}, the following statements hold for each arm.
\begin{enumerate}[label=(\roman*)]
\item For every $y$, $\tau\in\T$, and $e_2>e_1$,
\begin{equation}\label{eq:app-strong-monotonicity}
\psi_\tau(y,e_2)-\psi_\tau(y,e_1)
\le -2c_{\T}(e_2-e_1).
\end{equation}

\item The map $\tau\mapsto e(\tau)$ is Lipschitz on $\T$, while the map $\tau\mapsto D(\tau)$ is continuous and bounded away from zero.

\item The classes
\[
\Phi=\{\phi_\tau:\tau\in\T\},
\qquad
\Xi=\{\xi_\tau:\tau\in\T\},
\]
are $P$-Donsker, with envelope $C(1+|Y|)$. Their pairwise product classes
\[
\Phi^{(2)}
=
\{\phi_\tau\phi_s:\tau,s\in\T\},
\qquad
\Xi^{(2)}
=
\{\xi_\tau\xi_s:\tau,s\in\T\},
\]
are $P$-Glivenko--Cantelli. Furthermore, the maps $\tau\mapsto\phi_\tau$ and $\tau\mapsto\xi_\tau$ are continuous under the $L_2(P)$ semimetric.
\end{enumerate}
\end{lemma}

\begin{proof}
Part~(i) follows from \eqref{eq:app-strong-monotonicity-original}. Let $M(\tau,e)=\E\psi_\tau(Y,e)$. If $\tau_2>\tau_1$, strong monotonicity and $M\{\tau_j,e(\tau_j)\}=0$ give
\[
2c_{\T}|e(\tau_2)-e(\tau_1)|
\le 2|\tau_2-\tau_1|\E|Y-e(\tau_1)|.
\]
The expectation on the right-hand side is uniformly bounded over the compact interval $\T$, so $\tau\mapsto e(\tau)$ is Lipschitz. The continuity of $F$ near the expectile range under Assumption~\ref{ass:regularity} further implies that $D(\tau)=2[\tau+(1-2\tau)F\{e(\tau)\}]$ is continuous in $\tau$. In addition, $D(\tau)\ge2c_{\T}$. This proves part~(ii).

We next prove part~(iii). Let $K\subset\R$ be a compact interval containing $\{e(\tau):\tau\in\T\}$. Since $2c_{\T}\le D(\tau)\le2$, define the compact parameter set
\[
\Theta
=
\T\times K\times
\left[\frac12,\frac{1}{2c_{\T}}\right],
\]
and consider the class
\[
\mathcal H
=
\left\{
h_\theta(y)=a\psi_\tau(y,e):
\theta=(\tau,e,a)\in\Theta
\right\}.
\]
Because
\[
\phi_\tau(y)
=
D(\tau)^{-1}\psi_\tau\{y,e(\tau)\},
\]
the class $\Phi$ is a subclass of $\mathcal H$.

Uniformly over $(\tau,e)\in\T\times K$, the score satisfies $|\psi_\tau(y,e)|\le C(1+|y|)$ and
\[
|\psi_\tau(y,e)-\psi_s(y,e')|
\le
C(1+|y|)
\bigl(|\tau-s|+|e-e'|\bigr).
\]
Since $a$ ranges over a fixed compact interval, it follows that
\begin{equation}\label{eq:app-lipschitz-class}
|h_\theta(y)-h_{\theta'}(y)|
\le
C(1+|y|)\|\theta-\theta'\|,
\qquad
\theta,\theta'\in\Theta.
\end{equation}
Therefore, $\mathcal H$ is a finite-dimensional Lipschitz-parametric class with square-integrable envelope $C(1+|Y|)$. Since $\Theta$ is compact, Theorem~2.7.17 of \citet{vdvw2023} gives the required $L_2(P)$ bracketing bound, and the corresponding bracketing entropy integral is finite. By Theorem~2.5.6 of \citet{vdvw2023}, $\mathcal H$ is $P$-Donsker. Therefore, its subclass $\Phi$ is also $P$-Donsker.

The same argument applies to
\[
\mathcal K
=
\left\{
y\mapsto h_\theta(y)-(y-\mu):
\theta\in\Theta
\right\}.
\]
When differences between two elements of the class are taken, the fixed term $y-\mu$ cancels, so \eqref{eq:app-lipschitz-class} remains valid. The class $\mathcal K$ also has envelope $C(1+|y|)$. Since $\Xi\subset\mathcal K$, the class $\Xi$ is therefore $P$-Donsker.

For the product classes, let $\theta_1,\theta_2,\theta_1',\theta_2'\in\Theta$. The preceding envelope and Lipschitz bounds give
$$
|h_{\theta_1}(y)h_{\theta_2}(y) -h_{\theta_1'}(y)h_{\theta_2'}(y)|\le C(1+y^2) \left( \|\theta_1-\theta_1'\| + \|\theta_2-\theta_2'\| \right).
$$
Therefore, the product class generated by $\mathcal H$ is also a finite-dimensional Lipschitz-parametric class, now with integrable envelope $C(1+Y^2)$. Applying Theorem~2.7.17 of \citet{vdvw2023} with the $L_1(P)$ norm and parameter set $\Theta\times\Theta$ gives finite $L_1(P)$ bracketing numbers for every positive bracket size. By Theorem~2.4.1 of \citet{vdvw2023}, $\Phi^{(2)}$ is $P$-Glivenko--Cantelli. The same argument, with $\mathcal K$ in place of $\mathcal H$, gives the corresponding result for $\Xi^{(2)}$.

Finally, part~(ii) shows that both $e(\tau)$ and $D(\tau)$ are continuous in $\tau$, with $D(\tau)$ bounded away from zero. Therefore, for every fixed $y$, $\phi_\tau(y)$ is continuous in $\tau$. Since $|\phi_\tau(Y)|\le C(1+|Y|)$ uniformly over $\tau\in\T$ and this envelope is square integrable, the dominated convergence theorem gives
\[
\|\phi_\tau-\phi_s\|_{P,2}\to0
\qquad\text{as }s\to\tau.
\]
Since $\xi_\tau-\xi_s=\phi_\tau-\phi_s$, the same $L_2(P)$ continuity holds for $\Xi$.
\end{proof}
\paragraph{Proof of Lemma~\ref{lem:bahadur}}
\begin{proof}
Fix an arm and abbreviate $e_\tau=e(\tau)$ and $\widehat e_\tau=\widehat e(\tau)$. Define
\[
M_n(\tau,e)=\E_n\psi_\tau(Y,e),\qquad
w_\tau(y)=|\tau-\1\{y\le e_\tau\}|,\qquad
D_n(\tau)=2\E_n w_\tau(Y).
\]

Since $\E\phi_\tau=0$ and $M_n(\tau,e_\tau)=D(\tau)n^{-1/2}\mathbb G_n\phi_\tau$, part~(iii) of Lemma~\ref{lem:appendix-regularity} and the boundedness of $D$ imply
\[
\sup_{\tau\in\T}|M_n(\tau,e_\tau)|=O_p(n^{-1/2}).
\]
Taking the sample average of \eqref{eq:app-strong-monotonicity} gives, for $e_2>e_1$,
$M_n(\tau,e_2)-M_n(\tau,e_1)\le-2c_{\T}(e_2-e_1)$. Since $M_n(\tau,\widehat e_\tau)=0$, applying this inequality to $e_\tau$ and $\widehat e_\tau$ in the appropriate order gives
$2c_{\T}|\widehat e_\tau-e_\tau|\le|M_n(\tau,e_\tau)|$. Therefore,
\begin{equation}\label{eq:app-root-rate}
\sup_{\tau\in\T}|\widehat e_\tau-e_\tau|
=
O_p(n^{-1/2}).
\end{equation}

We next linearize the score function around $e_\tau$. Using the identity
\[
\psi_\tau(y,e)
=
2\{\tau+(1-2\tau)\1(y\le e)\}(y-e),
\]
we obtain the following decomposition:
\begin{equation}\label{eq:app-score-expansion}
\psi_\tau(y,e)
=
\psi_\tau(y,e_\tau)
-2w_\tau(y)(e-e_\tau)
+r_\tau(y,e),
\end{equation}
where $r_\tau(y,e)=2(1-2\tau)(y-e)\left[\1\{y\le e\}-\1\{y\le e_\tau\}\right]$. This remainder is nonzero only when $y$ lies between $e$ and $e_\tau$. On this event, $|y-e|\le|e-e_\tau|$, and therefore
\begin{equation}\label{eq:app-score-remainder-bound}
|r_\tau(y,e)|
\le
2|e-e_\tau|
\big|
\1\{y\le e\}-\1\{y\le e_\tau\}
\big|.
\end{equation}

Setting $e=\widehat e_\tau$ in \eqref{eq:app-score-expansion} and using $M_n(\tau,\widehat e_\tau)=0$ gives
\begin{align}
\widehat e_\tau-e_\tau
&=
\frac{M_n(\tau,e_\tau)}{D_n(\tau)}
+
\frac{\E_n r_\tau(Y,\widehat e_\tau)}{D_n(\tau)}
\notag\\
&=
\E_n\phi_\tau
+
M_n(\tau,e_\tau)
\left\{
\frac{1}{D_n(\tau)}-\frac{1}{D(\tau)}
\right\}
+
\frac{\E_n r_\tau(Y,\widehat e_\tau)}{D_n(\tau)}.
\label{eq:app-score-solve}
\end{align}
We control the last two terms on the right-hand side separately.

First, $w_\tau(y)\ge c_{\T}$ implies that $D_n(\tau)\ge2c_{\T}$ for every $n$ and $\tau$. In addition,
\[
D_n(\tau)-D(\tau)
=
2(1-2\tau)\{F_n(e_\tau)-F(e_\tau)\},
\]
so the Dvoretzky--Kiefer--Wolfowitz inequality gives
$\sup_{\tau\in\T}|D_n(\tau)-D(\tau)|=O_p(n^{-1/2})$. Since $D$ is also uniformly bounded away from zero, the second term in \eqref{eq:app-score-solve} satisfies
\[
\sup_{\tau\in\T}
|M_n(\tau,e_\tau)|
\left|
\frac{1}{D_n(\tau)}-\frac{1}{D(\tau)}
\right|
=
O_p(n^{-1}).
\]

We next control the crossing remainder. Let $\mathcal N$ be a compact neighborhood of the expectile range $\{e_\tau:\tau\in\T\}$ contained in the continuity neighborhood specified in Assumption~\ref{ass:regularity}. Define the modulus of continuity of $F$ by
\[
\omega_F(\delta)
=
\sup_{\substack{x,y\in\mathcal N\\ |x-y|\le\delta}}
|F(x)-F(y)|.
\]
Since $F$ is uniformly continuous on $\mathcal N$, $\omega_F(\delta)\to0$ as $\delta\downarrow0$. By \eqref{eq:app-root-rate}, $\widehat e_\tau\in\mathcal N$ uniformly over $\tau\in\T$ with probability tending to one. On this event, since the half-lines indexed by the two indicator functions are nested,
\[
\begin{aligned}
\sup_{\tau\in\T}
\E_n
\big|
\1\{Y\le\widehat e_\tau\}
-\1\{Y\le e_\tau\}
\big|
&=
\sup_{\tau\in\T}
|F_n(\widehat e_\tau)-F_n(e_\tau)| \\
&\le
2\|F_n-F\|_\infty
+
\omega_F\!\left(
\sup_{\tau\in\T}|\widehat e_\tau-e_\tau|
\right)
=
o_p(1).
\end{aligned}
\]
Combining this result with \eqref{eq:app-score-remainder-bound} and \eqref{eq:app-root-rate} gives
\[
\sup_{\tau\in\T}
\big|\E_n r_\tau(Y,\widehat e_\tau)\big|
\le
2\sup_{\tau\in\T}|\widehat e_\tau-e_\tau|
\sup_{\tau\in\T}
\E_n
\big|
\1\{Y\le\widehat e_\tau\}
-\1\{Y\le e_\tau\}
\big|
=o_p(n^{-1/2}).
\]

Therefore, the last two terms in \eqref{eq:app-score-solve} can be combined into a remainder $R_n(\tau)$ satisfying
\[
\widehat e(\tau)-e(\tau)
=
\E_n\phi_\tau+R_n(\tau),
\qquad
\sup_{\tau\in\T}|R_n(\tau)|=o_p(n^{-1/2}),
\]
which proves the uniform Bahadur representation.

If $F$ additionally has a bounded density on a neighborhood of the expectile range, then $\omega_F(\delta)\le C\delta$ for sufficiently small $\delta$. In this case, the empirical crossing mass above is $O_p(n^{-1/2})$, so
$\sup_{\tau\in\T}|\E_n r_\tau(Y,\widehat e_\tau)|=O_p(n^{-1})$. Combined with the denominator estimation error of order $O_p(n^{-1})$, this sharpens the uniform remainder to $O_p(n^{-1})$.

Finally, when $\tau=1/2$,
$\widehat e(1/2)-e(1/2)=\bar Y-\mu=\E_n\phi_{1/2}$, and therefore $R_n(1/2)=0$. The Donsker assertion follows directly from part~(iii) of Lemma~\ref{lem:appendix-regularity}. This completes the proof.
\end{proof}

\paragraph{Proof of Theorem~\ref{thm:centered-limit}}
\begin{proof}
Apply Lemma~\ref{lem:bahadur} separately to the two arms. For $d\in\{0,1\}$,
\[
\widehat e_d(\tau)-e_d(\tau)
=
\E_{n,d}\phi_{d,\tau}
+
R_{n,d}(\tau),
\qquad
\sup_{\tau\in\T}|R_{n,d}(\tau)|
=
o_p(n_d^{-1/2}).
\]
Since $\bar Y_d-\mu_d=\E_{n,d}(Y-\mu_d)$ and $\E_d\phi_{d,\tau}=0$, we obtain the following expansion uniformly in $\tau$:
\[
\{\widehat e_d(\tau)-e_d(\tau)\}
-
(\bar Y_d-\mu_d)
=
\frac{1}{\sqrt{n_d}}\,
\mathbb G_{n,d}\xi_{d,\tau}
+
R_{n,d}(\tau).
\]
Substituting the expansions for the two arms into the definition of $S_n$ gives
\[
\begin{aligned}
S_n(\tau)
&=
\sqrt N\,r(\tau)
+
\sqrt{\frac{N}{n_1}}\,
\mathbb G_{n,1}\xi_{1,\tau}
-
\sqrt{\frac{N}{n_0}}\,
\mathbb G_{n,0}\xi_{0,\tau}
+o_p(1) \\
&=
\sqrt N\,r(\tau)
+
\sqrt{1-\lambda_n}\,
\mathbb G_{n,1}\xi_{1,\tau}
-
\sqrt{\lambda_n}\,
\mathbb G_{n,0}\xi_{0,\tau}
+o_p(1),
\end{aligned}
\]
where the remainder is uniform over $\T$. This proves \eqref{eq:S-linear}.

Now suppose that $H_0$ holds. Then there exists some $\delta$ such that
$Y(1)\stackrel{d}{=}Y(0)+\delta$, and hence
$e_1(\tau)=e_0(\tau)+\delta$ and $\mu_1=\mu_0+\delta$. Moreover,
\[
\psi_\tau(y+\delta,e_0(\tau)+\delta)
=
\psi_\tau(y,e_0(\tau)),
\qquad
D_1(\tau)=D_0(\tau),
\]
and therefore $\xi_{1,\tau}(y+\delta)=\xi_{0,\tau}(y)$. Thus, the two function classes
$\{\xi_{1,\tau}:\tau\in\T\}$ and
$\{\xi_{0,\tau}:\tau\in\T\}$ have the same covariance kernel,
$\Gamma(\tau,\tau')=\E_d\{\xi_{d,\tau}\xi_{d,\tau'}\}$,
$d\in\{0,1\}$.

By Lemma~\ref{lem:appendix-regularity} and the independence of the two samples,
\[
\bigl(
\mathbb G_{n,1}\xi_{1,\cdot},
\mathbb G_{n,0}\xi_{0,\cdot}
\bigr)
\rightsquigarrow
\bigl(
\mathbb G_{\xi,1},
\mathbb G_{\xi,0}
\bigr)
\quad\text{in }\ell^\infty(\T)^2,
\]
where $\mathbb G_{\xi,1}$ and $\mathbb G_{\xi,0}$ are independent centered Gaussian processes, each with covariance kernel $\Gamma$. Under $H_0$, $r(\tau)=0$ for every $\tau$, and $\lambda_n\to\lambda\in(0,1)$. Therefore,
\[
S_n
\rightsquigarrow
\sqrt{1-\lambda}\,\mathbb G_{\xi,1}
-
\sqrt{\lambda}\,\mathbb G_{\xi,0}.
\]
The process on the right-hand side is centered Gaussian with covariance kernel
\[
(1-\lambda)\Gamma(\tau,\tau')
+
\lambda\Gamma(\tau,\tau')
=
\Gamma(\tau,\tau').
\]
It therefore has the same law as $\mathbb G_\xi$ in \eqref{eq:S-null-limit}, which proves the asserted weak convergence.

Finally, $\xi_{d,1/2}=0$ identically for both arms, so
$\mathbb G_\xi(1/2)=0$. This completes the proof.
\end{proof}

\paragraph{Proof of Corollary~\ref{cor:functionals}}
\begin{proof}
The supremum functional is Lipschitz continuous with respect to the uniform norm. The quadratic functional is also continuous because, for any bounded functions $f$ and $g$,
\[
\left|\int_\T f^2-\int_\T g^2\right|
\le |\T|\,\|f-g\|_\infty(\|f\|_\infty+\|g\|_\infty).
\]
The conclusion therefore follows from Theorem~\ref{thm:centered-limit} and the continuous mapping theorem.
\end{proof}

\paragraph{Proof of Proposition~\ref{prop:fixed}}
\begin{proof}
The stochastic term in \eqref{eq:S-linear} is $O_p(1)$ in $\ell^\infty(\T)$. Therefore, $S_n/\sqrt N$ converges uniformly in probability to $r$ on $\T$. Applying the two continuous functionals considered in the preceding proof yields the probability limits stated in the proposition. Since $r$ is continuous, $r\not\equiv0$ implies
$\sup_{\tau\in\T}|r(\tau)|>0$ and
$\int_\T r^2(\tau)\,d\tau>0$.
\end{proof}

The following lemma gives the regularity properties needed to extend the fixed-distribution expansion to the local path. Let $F_s$ denote the distribution of $T_s(\varepsilon)$, and let $e_s(\tau)$, $D_s(\tau)$, $\phi_{s,\tau}$, and $\xi_{s,\tau}$ denote the corresponding population objects. When $s=0$, they reduce to $e(\tau)$, $D(\tau)$, $\phi_\tau$, and $\xi_\tau$, respectively.

\begin{lemma}\label{lem:appendix-path}
Suppose that the baseline distribution $F$ satisfies Assumption~\ref{ass:regularity}, and that the function $b:\R\to\R$ is globally Lipschitz continuous and satisfies $\E b(\varepsilon)=0$. Then the following statements hold.

\begin{enumerate}[label=(\roman*)]
\item As $s\to0$,
\begin{equation}\label{eq:app-path-derivative}
\sup_{\tau\in\T}
\left|
e_s(\tau)-e(\tau)-s h_b(\tau)
\right|
=o(|s|).
\end{equation}

\item As $s\to0$,
\begin{equation}\label{eq:app-path-l2}
\sup_{\tau\in\T}
\left\|
\xi_{s,\tau}\{T_s(\varepsilon)\}
-\xi_\tau(\varepsilon)
\right\|_{F,2}
=o(1).
\end{equation}
Furthermore, there exists some $s_0>0$ such that the function class
\[
\left\{
x\mapsto\xi_{s,\tau}\{T_s(x)\}:
|s|\le s_0,\ \tau\in\T
\right\}
\]
is $F$-Donsker with envelope $C(1+|x|)$.

\item For the triangular array associated with any sequence $s_n\to0$, let $m_n\to\infty$ and suppose that $X_{n,1},\ldots,X_{n,m_n}$ are i.i.d. within each row with common distribution $F_{s_n}$. If $\widehat e_{s_n,m_n}(\tau)$ denotes the corresponding sample expectile, then
\begin{equation}\label{eq:app-path-bahadur}
\widehat e_{s_n,m_n}(\tau)-e_{s_n}(\tau)
=
\frac{1}{m_n}\sum_{i=1}^{m_n}
\phi_{s_n,\tau}(X_{n,i})
+
R_n(\tau),
\end{equation}
where
$\sup_{\tau\in\T}|R_n(\tau)|=o_p(m_n^{-1/2})$.
\end{enumerate}
\end{lemma}

\begin{proof}
Let $L_b<\infty$ satisfy $|b(x)-b(y)|\le L_b|x-y|$ for all $x,y\in\R$. Since $T_s(x)=x+s b(x)$,
\begin{equation}\label{eq:app-bilip}
(1-|s|L_b)|x-y|
\le
|T_s(x)-T_s(y)|
\le
(1+|s|L_b)|x-y|.
\end{equation}
Thus, whenever $|s|L_b<1$, $T_s$ is strictly increasing and invertible on $\R$, and its inverse is Lipschitz continuous with constant $(1-|s|L_b)^{-1}$.

We first prove part~(i). Define
$M_s(\tau,e)=\E\psi_\tau\{T_s(\varepsilon),e\}$.
By the same argument as in Lemma~\ref{lem:appendix-regularity}, for every $s$ and $\tau\in\T$, the map $e\mapsto M_s(\tau,e)$ satisfies the strong monotonicity bound with constant $2c_{\T}$ and is strictly decreasing. Since
$M_s\{\tau,e_s(\tau)\}=M_0\{\tau,e(\tau)\}=0$, we have
\[
2c_{\T}|e_s(\tau)-e(\tau)|
\le
\big|
M_s\{\tau,e(\tau)\}
-
M_s\{\tau,e_s(\tau)\}
\big|
=
\big|M_s\{\tau,e(\tau)\}-M_0\{\tau,e(\tau)\}\big|.
\]
For fixed $e$, the map $y\mapsto\psi_\tau(y,e)$ is Lipschitz continuous with constant at most $2$, uniformly in $\tau$. Therefore,
\[
\begin{aligned}
\big|
M_s\{\tau,e(\tau)\}
-
M_0\{\tau,e(\tau)\}
\big|
&\le
\E\left|
\psi_\tau\{T_s(\varepsilon),e(\tau)\}
-
\psi_\tau\{\varepsilon,e(\tau)\}
\right| \\
&\le
2\E|T_s(\varepsilon)-\varepsilon|
=
2|s|\E|b(\varepsilon)|.
\end{aligned}
\]
Since this upper bound does not depend on $\tau$, it follows that
\begin{equation}\label{eq:app-path-rate}
\sup_{\tau\in\T}|e_s(\tau)-e(\tau)|
=
O(|s|).
\end{equation}

To compare the perturbed and baseline score thresholds under the same distribution, define
$u_s(\tau)=T_s^{-1}\{e_s(\tau)\}$.
Since $T_s$ is strictly increasing,
$\1\{T_s(x)\le e_s(\tau)\}=\1\{x\le u_s(\tau)\}$.
Furthermore, by \eqref{eq:app-bilip},
\[
|u_s(\tau)-e(\tau)|
\le
\frac{|e_s(\tau)-T_s\{e(\tau)\}|}{1-|s|L_b}.
\]
Since
$T_s\{e(\tau)\}=e(\tau)+s b\{e(\tau)\}$,
\eqref{eq:app-path-rate} and compactness of the expectile range imply
\begin{equation}\label{eq:app-us-rate}
\sup_{\tau\in\T}|u_s(\tau)-e(\tau)|
=
O(|s|).
\end{equation}

Recall that
$w_\tau(x)=|\tau-\1\{x\le e(\tau)\}|$.
A direct expansion of the score gives
\begin{equation}
\psi_\tau\{T_s(x),e_s(\tau)\}
-
\psi_\tau\{x,e(\tau)\}
=
2w_\tau(x)
\left[
s b(x)-\{e_s(\tau)-e(\tau)\}
\right]
+
R_{s,\tau}(x),
\label{eq:app-path-score-expansion}
\end{equation}
where
\[
R_{s,\tau}(x)
=
2(1-2\tau)
\{T_s(x)-e_s(\tau)\}
\left[
\1\{x\le u_s(\tau)\}
-
\1\{x\le e(\tau)\}
\right].
\]
As in the proof of Lemma~\ref{lem:bahadur}, this remainder is nonzero only when $x$ lies between $e(\tau)$ and $u_s(\tau)$. On this event,
\[
|T_s(x)-e_s(\tau)|
=
|T_s(x)-T_s\{u_s(\tau)\}|
\le
C|u_s(\tau)-e(\tau)|
\le C|s|,
\]
where the last inequality follows from \eqref{eq:app-us-rate}. Therefore,
\[
|R_{s,\tau}(x)|
\le
C|s|
\big|
\1\{x\le u_s(\tau)\}
-
\1\{x\le e(\tau)\}
\big|.
\]
Since $u_s(\tau)$ converges uniformly to $e(\tau)$ and $F$ is uniformly continuous on a neighborhood of the expectile range,
\[
\sup_{\tau\in\T}
|F\{u_s(\tau)\}-F\{e(\tau)\}|
=o(1).
\]
It follows that
\[
\sup_{\tau\in\T}
\E|R_{s,\tau}(\varepsilon)|
\le
C|s|
\sup_{\tau\in\T}
|F\{u_s(\tau)\}-F\{e(\tau)\}|
=
o(|s|).
\]

Taking expectations in \eqref{eq:app-path-score-expansion}, the left-hand side vanishes because
\[
\E\psi_\tau\{T_s(\varepsilon),e_s(\tau)\}
=
\E\psi_\tau\{\varepsilon,e(\tau)\}
=
0.
\]
Hence,
\[
0
=
2s\,\E\{w_\tau(\varepsilon)b(\varepsilon)\}
-
2\E w_\tau(\varepsilon)\{e_s(\tau)-e(\tau)\}
+
\E R_{s,\tau}(\varepsilon).
\]
Since
$D(\tau)=2\E w_\tau(\varepsilon)$
and
$\sup_{\tau\in\T}\E|R_{s,\tau}(\varepsilon)|=o(|s|)$,
we obtain, uniformly in $\tau$,
\[
D(\tau)\{e_s(\tau)-e(\tau)\}
=
2s\,\E\{w_\tau(\varepsilon)b(\varepsilon)\}
+
o(|s|).
\]
Because $D$ is uniformly bounded away from zero,
\[
e_s(\tau)-e(\tau)
=
s\,\frac{2\E\{w_\tau(\varepsilon)b(\varepsilon)\}}{D(\tau)}
+
o(|s|)
=
s h_b(\tau)+o(|s|)
\]
uniformly over $\tau\in\T$. This proves
\eqref{eq:app-path-derivative}.

We next establish the $L_2$ stability in part~(ii). Since $T_s$ is strictly increasing and
$u_s(\tau)=T_s^{-1}\{e_s(\tau)\}$,
\[
F_s\{e_s(\tau)\}
=
\Pr\{T_s(\varepsilon)\le e_s(\tau)\}
=
F\{u_s(\tau)\}.
\]
Therefore,
\[
\begin{aligned}
|D_s(\tau)-D(\tau)|
&=
2|1-2\tau|
\left|
F\{u_s(\tau)\}-F\{e(\tau)\}
\right| \\
&\le
2\left|
F\{u_s(\tau)\}-F\{e(\tau)\}
\right|.
\end{aligned}
\]
By \eqref{eq:app-us-rate} and the uniform continuity of $F$ near the expectile range,
\begin{equation}\label{eq:app-Ds-convergence}
\sup_{\tau\in\T}|D_s(\tau)-D(\tau)|=o(1).
\end{equation}

The score expansion \eqref{eq:app-path-score-expansion}, together with
\eqref{eq:app-path-rate} and the preceding bound on $R_{s,\tau}$, also gives
\[
\sup_{\tau\in\T}
\left|
\psi_\tau\{T_s(x),e_s(\tau)\}
-
\psi_\tau\{x,e(\tau)\}
\right|
\le
2|s||b(x)|+C|s|
\le
C|s|(1+|x|),
\]
where the last inequality uses the fact that the Lipschitz function $b$ has at most linear growth.

Moreover,
$\E T_s(\varepsilon)=\E\varepsilon+s\E b(\varepsilon)=0$,
so the mean under $F_s$ remains zero. Thus,
\[
\xi_{s,\tau}\{T_s(x)\}
=
\frac{\psi_\tau\{T_s(x),e_s(\tau)\}}{D_s(\tau)}
-
T_s(x),
\qquad
\xi_\tau(x)
=
\frac{\psi_\tau\{x,e(\tau)\}}{D(\tau)}
-x.
\]
Their difference can be written as
\[
\begin{aligned}
\xi_{s,\tau}\{T_s(x)\}-\xi_\tau(x)
={}&
\frac{
\psi_\tau\{T_s(x),e_s(\tau)\}
-
\psi_\tau\{x,e(\tau)\}
}{D_s(\tau)}
\\
&+
\psi_\tau\{x,e(\tau)\}
\left\{
\frac{1}{D_s(\tau)}-\frac{1}{D(\tau)}
\right\}
-sb(x).
\end{aligned}
\]
Since $D_s(\tau)$ and $D(\tau)$ are both uniformly bounded away from zero, \eqref{eq:app-Ds-convergence} and the preceding bounds imply that there exists a sequence $a_s\to0$ such that
\[
\sup_{\tau\in\T}
\left|
\xi_{s,\tau}\{T_s(x)\}-\xi_\tau(x)
\right|
\le
a_s(1+|x|).
\]
Since $\E\varepsilon^2<\infty$,
\[
\sup_{\tau\in\T}
\left\|
\xi_{s,\tau}\{T_s(\varepsilon)\}-\xi_\tau(\varepsilon)
\right\|_{F,2}
\le
a_s\|1+|\varepsilon|\|_{F,2}
=o(1),
\]
which proves \eqref{eq:app-path-l2}.

We now establish the Donsker property in part~(ii). The class of interest is
\[
\mathcal F_\xi
=
\left\{
x\mapsto\xi_{s,\tau}\{T_s(x)\}:
|s|\le s_0,\ \tau\in\T
\right\}.
\]
Rather than dealing directly with the dependence of this class on the implicitly defined quantities $e_s(\tau)$ and $D_s(\tau)$, we embed it in a larger finite-dimensional parametric class. By \eqref{eq:app-path-rate}, we can choose $s_0>0$ sufficiently small and a compact interval $\mathcal K$ containing
\[
\{e_s(\tau):|s|\le s_0,\ \tau\in\T\}.
\]
Since
$2c_{\T}\le D_s(\tau)\le2$,
it follows that
$D_s(\tau)^{-1}\in[1/2,(2c_{\T})^{-1}]$.
Consider the function class
\[
\mathcal H
=
\left\{
x\mapsto
a\,\psi_\tau\{T_s(x),e\}
-x-sb(x):
|s|\le s_0,\ \tau\in\T,\ e\in\mathcal K,\
a\in
\left[\frac12,\frac{1}{2c_{\T}}\right]
\right\}.
\]
For $e=e_s(\tau)$ and $a=D_s(\tau)^{-1}$,
\[
a\,\psi_\tau\{T_s(x),e\}-x-sb(x)
=
\xi_{s,\tau}\{T_s(x)\},
\]
so $\mathcal F_\xi\subset\mathcal H$.

It therefore suffices to show that $\mathcal H$ is $F$-Donsker. Let
$\theta=(s,\tau,e,a)$, and denote the corresponding function by $h_\theta$. The Lipschitz properties of the score function and $b$ imply that, for any $\theta$ and $\theta'$ in the compact parameter set,
\[
|h_\theta(x)-h_{\theta'}(x)|
\le
C(1+|x|)\|\theta-\theta'\|.
\]
Furthermore, since $b$ has at most linear growth and both $e$ and $a$ range over compact sets,
\[
\sup_\theta |h_\theta(x)|
\le
C(1+|x|).
\]
Consequently, $\mathcal H$ is a finite-dimensional Lipschitz-parametric class with square-integrable envelope $C(1+|x|)$. The same bracketing argument as in Lemma~\ref{lem:appendix-regularity} therefore shows that $\mathcal H$ is $F$-Donsker. Since $\mathcal F_\xi$ is a subclass of $\mathcal H$, the desired Donsker property follows.

It remains to establish the triangular-array expansion in part~(iii). Represent the $n$th row as
$X_{n,i}=T_{s_n}(\varepsilon_i)$, where
$\varepsilon_i\stackrel{\mathrm{i.i.d.}}{\sim}F$.
By the same finite-dimensional bracketing argument as in part~(iii) of Lemma~\ref{lem:appendix-regularity}, the function class
\[
\left\{
x\mapsto\phi_{s,\tau}\{T_s(x)\}:
|s|\le s_0,\ \tau\in\T
\right\}
\]
is contained in a fixed $F$-Donsker class with envelope
$C(1+|x|)$. Since
$\E\phi_{s,\tau}\{T_s(\varepsilon)\}=0$
and $D_s(\tau)$ is uniformly bounded, the empirical score evaluated at
$e_{s_n}(\tau)$ is uniformly $O_p(m_n^{-1/2})$ over $\tau\in\T$.

Moreover, averaging \eqref{eq:app-strong-monotonicity} over the observations in the $n$th row gives the same uniform slope bound for the empirical score. Since
$\widehat e_{s_n,m_n}(\tau)$ is a zero of this empirical score, the argument used to prove \eqref{eq:app-root-rate} in Lemma~\ref{lem:bahadur} yields
\[
\sup_{\tau\in\T}
|\widehat e_{s_n,m_n}(\tau)-e_{s_n}(\tau)|
=
O_p(m_n^{-1/2}).
\]

It remains to verify that the threshold-crossing argument in Lemma~\ref{lem:bahadur} remains valid along the path. Let $F_s$ denote the CDF of
$T_s(\varepsilon)$. Since $T_s$ is strictly increasing,
$F_s(y)=F\{T_s^{-1}(y)\}$.
Furthermore, by \eqref{eq:app-bilip}, for all sufficiently small $s$,
\[
\begin{aligned}
|F_s(y)-F_s(z)|
&=
\left|
F\{T_s^{-1}(y)\}-F\{T_s^{-1}(z)\}
\right| \\
&\le
\omega_F\!\left(
|T_s^{-1}(y)-T_s^{-1}(z)|
\right)
\le
\omega_F(C|y-z|),
\end{aligned}
\]
where $\omega_F$ is a modulus of continuity of $F$ near the baseline expectile range, and the constant $C$ does not depend on $s$.

Let $\widehat F_{n,s_n}$ denote the empirical CDF of the $n$th row. Then
\[
\begin{aligned}
\sup_{\tau\in\T}
\big|
\widehat F_{n,s_n}\{\widehat e_{s_n,m_n}(\tau)\}
-
\widehat F_{n,s_n}\{e_{s_n}(\tau)\}
\big|
&\le
2\|\widehat F_{n,s_n}-F_{s_n}\|_\infty \\
&\quad+
\omega_F\!\left(
C\sup_{\tau\in\T}
|\widehat e_{s_n,m_n}(\tau)-e_{s_n}(\tau)|
\right)
=o_p(1),
\end{aligned}
\]
where the first term is $O_p(m_n^{-1/2})$ by the Dvoretzky--Kiefer--Wolfowitz inequality. Hence, by the same calculation as in Lemma~\ref{lem:bahadur}, the threshold-crossing remainder is $o_p(m_n^{-1/2})$. Combining this result with
$D_{n,s_n}(\tau)\ge2c_{\T}$
and the corresponding
$O_p(m_n^{-1/2})$
denominator approximation yields
\eqref{eq:app-path-bahadur} and completes the proof.
\end{proof}
\paragraph{Proof of Theorem~\ref{thm:local}}
\begin{proof}
By translation equivariance, it suffices to consider the innovation distributions. Since
$\E T_s(\varepsilon)=\E\varepsilon+s\E b(\varepsilon)=0$,
the mean remains unchanged along the path, and therefore
\[
r_n(\tau)=e_{s_n}(\tau)-e(\tau).
\]
Applying part~(i) of Lemma~\ref{lem:appendix-path} with
$s_n=\kappa/\sqrt N$
gives
\[
\sup_{\tau\in\T}
\left|
\sqrt N\,r_n(\tau)-\kappa h_b(\tau)
\right|
\to0,
\]
which proves \eqref{eq:path-drift}.

We next consider the stochastic component. Combining the triangular-array Bahadur representation for the treatment group in Lemma~\ref{lem:appendix-path} with the representation for the control group in Lemma~\ref{lem:bahadur}, and subtracting the corresponding sample mean expansions, gives, uniformly over $\tau\in\T$,
\[
\begin{aligned}
S_n(\tau)
={}&
\sqrt N\,r_n(\tau)
+
\sqrt{1-\lambda_n}\,
\mathbb G_{n,1}
\big[
\xi_{s_n,\tau}\{T_{s_n}(\cdot)\}
\big] \\
&\quad-
\sqrt{\lambda_n}\,
\mathbb G_{n,0}\xi_\tau
+o_p(1),
\end{aligned}
\]
where $\mathbb G_{n,1}$ and $\mathbb G_{n,0}$ are the empirical processes generated by two independent samples from the baseline distribution $F$. By \eqref{eq:app-path-l2} and the Donsker property in Lemma~\ref{lem:appendix-path}, asymptotic equicontinuity gives
\[
\sup_{\tau\in\T}
\left|
\mathbb G_{n,1}
\left[
\xi_{s_n,\tau}\{T_{s_n}(\cdot)\}-\xi_\tau
\right]
\right|
=o_p(1).
\]
Consequently,
\[
S_n(\tau)
=
\sqrt N\,r_n(\tau)
+
\sqrt{1-\lambda_n}\,
\mathbb G_{n,1}\xi_\tau
-
\sqrt{\lambda_n}\,
\mathbb G_{n,0}\xi_\tau
+o_p(1)
\]
uniformly over $\T$. The two empirical processes are independent and converge to independent copies of the same centered Gaussian process with covariance kernel
$\E\{\xi_\tau(\varepsilon)\xi_{\tau'}(\varepsilon)\}$.
Since $\lambda_n\to\lambda\in(0,1)$, their weighted difference has the same limiting covariance kernel and therefore converges to the null Gaussian process
$\mathbb G_\xi$
in Theorem~\ref{thm:centered-limit}. Combining this result with \eqref{eq:path-drift} yields
\[
S_n\rightsquigarrow\mathbb G_\xi+\kappa h_b
\qquad\text{in }\ell^\infty(\T),
\]
which proves \eqref{eq:S-local}.

It remains to establish the sign result. Define
$B_b(z)=\E[b(\varepsilon)\1\{\varepsilon\le z\}]$.
When $z=e(\tau)$, both sides of the threshold have positive probability. Since $b$ is nondecreasing and
$\Var\{b(\varepsilon)\}>0$,
the conditional mean of $b(\varepsilon)$ below the threshold is strictly smaller than its conditional mean above the threshold. Together with
$\E b(\varepsilon)=0$,
this implies
$B_b\{e(\tau)\}<0$.
Moreover,
\[
\E\{w_\tau(\varepsilon)b(\varepsilon)\}
=
(1-\tau)B_b\{e(\tau)\}
+\tau\big[\E b(\varepsilon)-B_b\{e(\tau)\}\big]
=
(1-2\tau)B_b\{e(\tau)\}.
\]
Therefore,
\[
h_b(\tau)
=
\frac{2(1-2\tau)B_b\{e(\tau)\}}{D(\tau)}.
\]
Since $D(\tau)>0$, we have
$h_b(\tau)<0$ for $\tau<1/2$,
$h_b(\tau)>0$ for $\tau>1/2$, and
$h_b(1/2)=0$.
\end{proof}

\subsection{Proofs for Section~\ref{sec:inference}}
\label{app:inference}

\paragraph{Proof of Theorem~\ref{thm:multiplier-fclt}}
\begin{proof}
We first establish the asymptotic equivalence between the feasible multiplier process and the oracle multiplier process, and then determine the limit of the oracle process. Let $\E_V$ and $P^*$ denote expectation and probability, respectively, with respect to the multipliers conditional on the observed data.

\emph{Step 1: Approximation of the feasible process by the oracle process.}
By definition,
\begin{align*}
\widehat\xi_{d,\tau}(y)-\xi_{d,\tau}(y)
={}&
\frac{
\psi_\tau\{y,\widehat e_d(\tau)\}
-
\psi_\tau\{y,e_d(\tau)\}
}{
\widehat D_d(\tau)
} \\
&+
\psi_\tau\{y,e_d(\tau)\}
\left\{
\frac{1}{\widehat D_d(\tau)}
-
\frac{1}{D_d(\tau)}
\right\}
+
(\bar Y_d-\mu_d).
\end{align*}
Using
$|\psi_\tau(y,e)-\psi_\tau(y,e')|\le2|e-e'|$,
$\sup_{\tau\in\T}|\psi_\tau\{y,e_d(\tau)\}|\le C(1+|y|)$,
and
$2c_{\T}\le\widehat D_d(\tau),D_d(\tau)\le2$,
we obtain
\begin{align*}
\sup_{\tau\in\T}
|\widehat\xi_{d,\tau}(y)-\xi_{d,\tau}(y)|
\le C\bigg\{&
\sup_{\tau\in\T}
|\widehat e_d(\tau)-e_d(\tau)| \\
&+
(1+|y|)
\sup_{\tau\in\T}
|\widehat D_d(\tau)-D_d(\tau)|
+
|\bar Y_d-\mu_d|
\bigg\}.
\end{align*}
For fixed arm distributions, Lemma~\ref{lem:bahadur}, continuity of the CDFs near the expectile ranges, and the Dvoretzky--Kiefer--Wolfowitz inequality give
\[
\sup_{\tau\in\T}|\widehat e_d(\tau)-e_d(\tau)|
=
O_p(n_d^{-1/2}),
\qquad
\sup_{\tau\in\T}|\widehat D_d(\tau)-D_d(\tau)|
=
o_p(1).
\]
Together with
$\bar Y_d-\mu_d=O_p(n_d^{-1/2})$,
these bounds imply that the right-hand side is bounded by
$a_{d,n}(1+|y|)$
for some
$a_{d,n}=o_p(1)$.
Since
$\E_{n,d}(1+|Y|)^2=O_p(1)$,
it follows that
\begin{equation}\label{eq:app-xi-l2}
\delta_{d,n}
:=
\sup_{\tau\in\T}
\|\widehat\xi_{d,\tau}-\xi_{d,\tau}\|_{\E_{n,d},2}
=
o_p(1).
\end{equation}

Since the mapping
$\tau\mapsto e_d(\tau)$
is continuous on the compact set $\T$, its range is also compact. Therefore, we may choose fixed compact intervals
$\mathcal K_d$ and $\mathcal M_d$ such that
\[
\{e_d(\tau):\tau\in\T\}
\subset
\operatorname{int}(\mathcal K_d),
\qquad
\mu_d\in\operatorname{int}(\mathcal M_d),
\]
and let
$\mathcal A=[1/2,1/(2c_{\T})]$.
Define the function class
\[
\mathcal H_d
=
\left\{
y\mapsto a\psi_\tau(y,e)-y+m:
(\tau,e,a,m)\in
\T\times\mathcal K_d\times\mathcal A\times\mathcal M_d
\right\}.
\]
Since
\[
\xi_{d,\tau}(y)
=
D_d(\tau)^{-1}\psi_\tau\{y,e_d(\tau)\}-y+\mu_d,
\]
and
$e_d(\tau)\in\mathcal K_d$,
$D_d(\tau)^{-1}\in\mathcal A$,
and
$\mu_d\in\mathcal M_d$,
we have
$\{\xi_{d,\tau}:\tau\in\T\}\subset\mathcal H_d$.
Moreover,
\[
\widehat\xi_{d,\tau}(y)
=
\widehat D_d(\tau)^{-1}
\psi_\tau\{y,\widehat e_d(\tau)\}
-y+\bar Y_d.
\]
Since
$\sup_{\tau\in\T}|\widehat e_d(\tau)-e_d(\tau)|=o_p(1)$,
$\bar Y_d-\mu_d=o_p(1)$,
and
$\widehat D_d(\tau)^{-1}\in\mathcal A$
for all
$\tau\in\T$,
it follows that
\[
P\left(
\{\widehat\xi_{d,\tau}:\tau\in\T\}\subset\mathcal H_d
\right)\to1.
\]

As in Lemma~\ref{lem:appendix-regularity},
$\mathcal H_d$
is a finite-dimensional Lipschitz-parametric class with envelope
$L_d(y)=C(1+|y|)$.
By
\citet[Theorem~2.7.17]{vdvw2023}
and the polynomial covering numbers of its compact finite-dimensional parameter space,
$\mathcal H_d-\mathcal H_d$
has polynomial packing numbers under
$L_2(\E_{n,d})$
after normalization by
$\|L_d\|_{\E_{n,d},2}$,
uniformly in $n_d$ and the realized sample.

Let
\[
\mathcal G_{d,n}
=
\left\{
g\in\mathcal H_d-\mathcal H_d:
\|g\|_{\E_{n,d},2}\le\delta_{d,n}
\right\}.
\]
On the event of the preceding class inclusion,
\[
\widehat\xi_{d,\tau}-\xi_{d,\tau}\in\mathcal G_{d,n}
\qquad\text{for every }\tau\in\T,
\]
and
$0\in\mathcal G_{d,n}$.
Conditional on the data, define
\[
\mathbb M_{d,n}^*g
=
\frac1{\sqrt{n_d}}
\sum_{i=1}^{n_d}V_i^{(d)}g(Y_i^{(d)}),
\qquad
g\in\mathcal G_{d,n}.
\]
Since the multipliers are centered and bounded,
$\E_V V_i^{(d)}=0$
and
$|V_i^{(d)}|\le C$.
Therefore, for any
$g,h\in\mathcal G_{d,n}$
and
$t\in\mathbb R$,
Hoeffding's lemma gives
\begin{align*}
\E_V\exp\left[
t\{\mathbb M_{d,n}^*g-\mathbb M_{d,n}^*h\}
\right]
&=
\prod_{i=1}^{n_d}
\E_V\exp\left[
\frac{tV_i^{(d)}}{\sqrt{n_d}}
\{g-h\}(Y_i^{(d)})
\right] \\
&\le
\exp\left\{
\frac{Ct^2}{n_d}
\sum_{i=1}^{n_d}
\{g-h\}^2(Y_i^{(d)})
\right\} \\
&=
\exp\left\{
Ct^2
\|g-h\|_{\E_{n,d},2}^2
\right\}.
\end{align*}
Thus,
$\{\mathbb M_{d,n}^*g:g\in\mathcal G_{d,n}\}$
is sub-Gaussian with respect to a constant multiple of the
$L_2(\E_{n,d})$
semimetric. Since
$\|g\|_{\E_{n,d},2}\le\delta_{d,n}$
for every
$g\in\mathcal G_{d,n}$
and
$0\in\mathcal G_{d,n}$,
\citet[Corollary~2.2.9]{vdvw2023}
yields
\begin{align*}
&\E_V\sup_{\tau\in\T}
\left|
\frac1{\sqrt{n_d}}\sum_{i=1}^{n_d}V_i^{(d)}
\left\{
\widehat\xi_{d,i,\tau}
-\xi_{d,\tau}(Y_i^{(d)})
\right\}
\right| \\
&\quad\le
\E_V\sup_{g\in\mathcal G_{d,n}}
|\mathbb M_{d,n}^*g| \\
&\quad\le
C\int_0^{C\delta_{d,n}}
\sqrt{
\log D\!\left(
u,\mathcal G_{d,n},L_2(\E_{n,d})
\right)
}\,du .
\end{align*}
The polynomial packing bound for
$\mathcal H_d-\mathcal H_d$
implies that, for some fixed
$q<\infty$,
\[
D\!\left(
u,\mathcal G_{d,n},L_2(\E_{n,d})
\right)
\le
C\left(
1+\frac{C\|L_d\|_{\E_{n,d},2}}{u}
\right)^q .
\]
Consequently,
\begin{equation*}
\E_V\sup_{\tau\in\T}
\left|
\frac1{\sqrt{n_d}}\sum_{i=1}^{n_d}V_i^{(d)}
\left\{
\widehat\xi_{d,i,\tau}
-\xi_{d,\tau}(Y_i^{(d)})
\right\}
\right|
\le
C\delta_{d,n}
\left[
1+
\sqrt{
\log\left(
1+\frac{\|L_d\|_{\E_{n,d},2}}{\delta_{d,n}}
\right)}
\right]
=o_p(1),
\end{equation*}
where the upper bound is interpreted as zero when
$\delta_{d,n}=0$.
Indeed,
$\delta_{d,n}=o_p(1)$,
$\|L_d\|_{\E_{n,d},2}=O_p(1)$,
and
$x\sqrt{\log(1+M/x)}\to0$
as
$x\downarrow0$
for every
$M<\infty$.
Conditional Markov's inequality therefore yields
\begin{equation}\label{eq:app-multiplier-negligible}
\sup_{\tau\in\T}
\left|
\frac1{\sqrt{n_d}}\sum_{i=1}^{n_d}V_i^{(d)}
\{\widehat\xi_{d,i,\tau}-\xi_{d,\tau}(Y_i^{(d)})\}
\right|
=o_{P^*}(1)
\end{equation}
in probability.

The preceding feasible-to-oracle argument applies to any fixed arm distribution satisfying the stated regularity conditions and therefore does not depend on $H_0$. Under
$H_{1n}(b,\kappa)$,
the treatment distribution varies with $n$.
Parts~(ii)--(iii) of Lemma~\ref{lem:appendix-path}, together with the uniform bounds established there, ensure that the same empirical-$L_2$ approximation and packing bounds hold uniformly along
$s_n\to0$.
Consequently,
\eqref{eq:app-multiplier-negligible}
continues to hold under the local alternatives.

\emph{Step 2: Oracle limits.}
For fixed arm distributions, Lemma~\ref{lem:appendix-regularity} shows that
$\{\xi_{d,\tau}:\tau\in\T\}$
is an
$F_d$-Donsker class, with
$\E_d\xi_{d,\tau}=0$.
Since the multipliers are independent of the data, centered, have unit variance, and are bounded, they satisfy the
$L_{2,1}$
condition in
\citet[Theorem~2.9.6]{vdvw2023}.
Therefore, for each arm,
\[
\frac1{\sqrt{n_d}}\sum_{i=1}^{n_d}
V_i^{(d)}\xi_{d,\cdot}(Y_i^{(d)})
\rightsquigarrow^*
\mathbb G_{\xi,d}
\qquad
\text{in probability in }\ell^\infty(\T),
\]
where
$\mathbb G_{\xi,d}$
is a centered Gaussian process with covariance kernel
$\E_d\{\xi_{d,\tau}\xi_{d,\tau'}\}$.

Since the two multiplier sequences are conditionally independent,
$\lambda_n\to\lambda$,
and
\eqref{eq:app-multiplier-negligible}
holds, the conditional Slutsky theorem implies that the feasible two-sample multiplier process converges to a centered Gaussian process with covariance kernel
\[
(1-\lambda)\E_1\{\xi_{1,\tau}\xi_{1,\tau'}\}
+
\lambda\E_0\{\xi_{0,\tau}\xi_{0,\tau'}\}.
\]
Under $H_0$, translation equivariance makes the two arm-specific covariance kernels equal, so the expression reduces to
$\Gamma(\tau,\tau')$.
Under a fixed alternative, the covariance kernel equals
$\Gamma_A(\tau,\tau')$.

Along the local alternatives
$H_{1n}(b,\kappa)$,
let
$f_{s,\tau}(x)=\xi_{s,\tau}\{T_s(x)\}$.
By Lemma~\ref{lem:appendix-path},
$\{f_{s,\tau}:|s|\le s_0,\tau\in\T\}$
is contained in a common
$F$-Donsker class,
$\E_F f_{s,\tau}=0$,
and
\[
\sup_{\tau\in\T}
\|f_{s_n,\tau}-\xi_\tau\|_{F,2}
=o(1).
\]
The conditional asymptotic equicontinuity implied by
\citet[Theorem~2.9.6]{vdvw2023}
therefore gives
\[
\sup_{\tau\in\T}
\left|
\frac1{\sqrt{n_1}}
\sum_{i=1}^{n_1}V_i^{(1)}
\{f_{s_n,\tau}(\varepsilon_i^{(1)})
-\xi_\tau(\varepsilon_i^{(1)})\}
\right|
=o_{P^*}(1)
\]
in probability.
Hence, the local oracle treatment process may be replaced by its baseline counterpart. Together with the unchanged control process, this implies that the oracle two-sample process converges to
$\mathbb G_\xi$.

Combining Steps~1 and~2,
\[
S_n^*
=
S_n^{o,*}+o_{P^*}(1)
\rightsquigarrow^*
\begin{cases}
\mathbb G_\xi,
& \text{under } H_0 \text{ and } H_{1n}(b,\kappa),\\[2mm]
\mathbb G_\xi^A,
& \text{under a fixed alternative},
\end{cases}
\qquad
\text{in probability in }\ell^\infty(\T),
\]
where
$S_n^{o,*}$
denotes the corresponding oracle two-sample multiplier process.
\end{proof}

\paragraph{Proof of Proposition~\ref{prop:factorization}}
\begin{proof}
We suppress the arm subscript throughout. From
\[
\psi_\tau(y,e)
=
(y-e)+c_\tau|y-e|,
\]
the population and sample score equations give
\[
0=\mu-e(\tau)+c_\tau m(\tau),
\qquad
0=\bar Y-\widehat e(\tau)+c_\tau\widehat m(\tau),
\]
and therefore
\begin{equation*}
e(\tau)-\mu=c_\tau m(\tau),
\qquad
\widehat e(\tau)-\bar Y=c_\tau\widehat m(\tau).
\end{equation*}

We next linearize $\widehat m(\tau)$. Using the identity
\[
\begin{aligned}
|y-\widehat e(\tau)|-|y-e(\tau)|
={}&
\{\widehat e(\tau)-e(\tau)\}
\{2\1(y\le e(\tau))-1\} \\
&+
2\{\widehat e(\tau)-y\}
\{\1(y\le\widehat e(\tau))-\1(y\le e(\tau))\},
\end{aligned}
\]
together with the crossing-mass argument in Lemma~\ref{lem:bahadur}, we obtain, uniformly over $\tau\in\T$,
\begin{align*}
\widehat m(\tau)-m(\tau)
={}&
(\E_n-\E)|Y-e(\tau)| \\
&+
\{\widehat e(\tau)-e(\tau)\}
\{2\E_n\1(Y\le e(\tau))-1\}
+o_p(n^{-1/2}) \\
={}&
(\E_n-\E)|Y-e(\tau)|
-a(\tau)\{\widehat e(\tau)-e(\tau)\}
+o_p(n^{-1/2}),
\end{align*}
where the second equality follows from the Dvoretzky--Kiefer--Wolfowitz inequality,
$2F\{e(\tau)\}-1=-a(\tau)$,
and
$\sup_{\tau\in\T}|\widehat e(\tau)-e(\tau)|=O_p(n^{-1/2})$.
Therefore, Lemma~\ref{lem:bahadur} gives
\begin{align*}
\widehat m(\tau)-m(\tau)
&=
\E_n
\left[
|Y-e(\tau)|-m(\tau)-a(\tau)\phi_\tau(Y)
\right]
+o_p(n^{-1/2}) \\
&=
\E_n\chi_\tau+o_p(n^{-1/2}),
\end{align*}
where we have used
$D(\tau)=1+c_\tau a(\tau)$
and the score equation identity
$e(\tau)-\mu=c_\tau m(\tau)$.
Together, these identities imply
\[
|y-e(\tau)|-m(\tau)-a(\tau)\phi_\tau(y)
=
\frac{
|y-e(\tau)|-m(\tau)-a(\tau)(y-\mu)
}{
D(\tau)
}
=
\chi_\tau(y).
\]
This proves \eqref{eq:m-linearization}.

Finally,
\begin{align*}
\xi_\tau(y)
&=
\frac{y-e(\tau)+c_\tau|y-e(\tau)|}{D(\tau)}
-(y-\mu) \\
&=
\frac{
c_\tau\{|y-e(\tau)|-m(\tau)-a(\tau)(y-\mu)\}
}{
D(\tau)
}
=
c_\tau\chi_\tau(y),
\end{align*}
where the second equality again follows from the score equation identity and
$D(\tau)=1+c_\tau a(\tau)$.
Similarly,
\[
\widehat e(\tau)-\bar Y=c_\tau\widehat m(\tau),
\qquad
\widehat D(\tau)=1+c_\tau\widehat a(\tau)
\]
imply that
$\widehat\xi_{i,\tau}
=
c_\tau\widehat\chi_{i,\tau}$.
In addition,
\[
\E_n\widehat\chi_\tau
=
\frac{
\E_n|Y-\widehat e(\tau)|
-\widehat m(\tau)
-\widehat a(\tau)\E_n(Y-\bar Y)
}{
\widehat D(\tau)
}
=0.
\]
Applying these identities separately to the two arms yields
\[
r(\tau)=c_\tau q(\tau),
\qquad
S_n(\tau)=c_\tau A_n(\tau).
\]
\end{proof}

\paragraph{Proof of Corollary~\ref{cor:studentized}}
\begin{proof}
Using the same finite-dimensional Lipschitz-parametric argument as in Lemma~\ref{lem:appendix-regularity}, we obtain that
$\{\chi_{d,\tau}:\tau\in\T\}$
is an $F_d$-Donsker class with envelope
$C(1+|Y|)$.
Applying \eqref{eq:m-linearization} to each arm gives, uniformly over $\T$,
\begin{equation}\label{eq:app-A-linearization}
A_n(\tau)
=
\sqrt N\,q(\tau)
+
\sqrt{1-\lambda_n}\,\mathbb G_{n,1}\chi_{1,\tau}
-
\sqrt{\lambda_n}\,\mathbb G_{n,0}\chi_{0,\tau}
+
o_p(1).
\end{equation}
Under $H_0$,
$q(\tau)=0$,
while translation equivariance gives
$\chi_{1,\tau}(Y(1))\stackrel{d}{=}\chi_{0,\tau}(Y(0))$.
Therefore, by independence of the two samples and
$\lambda_n\to\lambda$,
\[
A_n\rightsquigarrow\mathbb G_\chi
\qquad\text{in }\ell^\infty(\T),
\]
where
$\mathbb G_\chi$
has covariance kernel
$K(\tau,\tau')$.

For studentization, uniform consistency of
$\widehat e_d$,
$\widehat m_d$,
$\widehat a_d$,
and
$\widehat D_d$,
together with
$\inf_{\tau\in\T}D_d(\tau)>0$,
gives
\begin{equation}\label{eq:app-chi-l2}
\sup_{\tau\in\T}
\|\widehat\chi_{d,\tau}-\chi_{d,\tau}\|_{\E_{n,d},2}
=
o_p(1).
\end{equation}
Moreover,
$\{\chi_{d,\tau}^2:\tau\in\T\}$
is a Glivenko--Cantelli class with envelope
$C(1+Y^2)$.
Therefore, by \eqref{eq:app-chi-l2} and the Cauchy--Schwarz inequality,
\[
\sup_{\tau\in\T}
|\widehat v_{A,n}^2(\tau)-v_A^2(\tau)|
=
o_p(1).
\]
Since
$\inf_{\tau\in\T}v_A^2(\tau)>0$,
Slutsky's theorem gives
\[
U_n
\rightsquigarrow
\frac{\mathbb G_\chi}{v_A}
=
\mathbb U
\qquad\text{in }\ell^\infty(\T).
\]
The limiting distributions of
$KS_n^{\mathrm{stud}}$
and
$CvM_n^{\mathrm{stud}}$
then follow from the continuous mapping theorem.
\end{proof}

\paragraph{Local limit of the studentized process.}
We verify the local limit stated after Corollary~\ref{cor:studentized}. Applying the absolute-value expansion from Proposition~\ref{prop:factorization} along the path $T_s$ gives, uniformly over $\T$,
\[
m_s(\tau)
=
m(\tau)+s g_b(\tau)+o(|s|),
\]
where, using
$\E b(\varepsilon)=0$,
\eqref{eq:hb-def},
and
$D(\tau)=1+c_\tau a(\tau)$,
\begin{align*}
g_b(\tau)
&=
-2\E\!\left[
b(\varepsilon)\1\{\varepsilon\le e(\tau)\}
\right]
-a(\tau)h_b(\tau) \\
&=
-\frac{
2\E\!\left[
b(\varepsilon)\1\{\varepsilon\le e(\tau)\}
\right]
}{
D(\tau)
}.
\end{align*}
In particular,
$h_b(\tau)=c_\tau g_b(\tau)$.
The same relation also follows by differentiating
$e_s(\tau)-\E T_s(\varepsilon)=c_\tau m_s(\tau)$
with respect to $s$.

The same pathwise argument gives
\[
\sup_{\tau\in\T}
\left\|
\chi_{s_n,\tau}\{T_{s_n}(\varepsilon)\}
-\chi_\tau(\varepsilon)
\right\|_{F,2}
=
o(1),
\]
with all these functions contained in a common
$F$-Donsker class.
Therefore, when
$s_n=\kappa/\sqrt N$,
\[
\sup_{\tau\in\T}
\left|
\sqrt N\,q_n(\tau)-\kappa g_b(\tau)
\right|
\to0,
\qquad
A_n
\rightsquigarrow
\mathbb G_\chi+\kappa g_b
\quad\text{in }\ell^\infty(\T).
\]
The variance consistency argument in \eqref{eq:app-chi-l2} continues to hold uniformly along
$s_n\to0$,
so
\[
\sup_{\tau\in\T}
|\widehat v_{A,n}(\tau)-v_A(\tau)|
=
o_p(1).
\]
Slutsky's theorem therefore gives
\[
U_n
\rightsquigarrow
\frac{\mathbb G_\chi+\kappa g_b}{v_A}
\qquad
\text{in }\ell^\infty(\T).
\]

\paragraph{Proof of Theorem~\ref{thm:covariate-adjustment}}
\begin{proof}
Write
$\widetilde Z=Z-\E Z$.
Since
$B_n=\sqrt N(\bar Z_1-\bar Z_0)$,
\begin{equation}\label{eq:app-B-linearization}
B_n
=
\sqrt{1-\lambda_n}\,\mathbb G_{n,1}\widetilde Z
-
\sqrt{\lambda_n}\,\mathbb G_{n,0}\widetilde Z.
\end{equation}
Therefore,
\[
B_n=O_p(1),
\qquad
\Var(B_n)\to\Sigma_Z.
\]

We first establish the uniform consistency of
$\widehat\gamma$.
By the law of large numbers,
$\widehat\Sigma_Z\xrightarrow{p}\Sigma_Z$.
Moreover, since
$\E_{n,d}\widehat\xi_{d,\tau}=0$,
\[
\E_{n,d}\{(Z-\bar Z_d)\widehat\xi_{d,\tau}\}
=
\E_{n,d}\{\widetilde Z\,\widehat\xi_{d,\tau}\}.
\]
For each arm,
\begin{align*}
&\sup_{\tau\in\T}
\left\|
\E_{n,d}\{\widetilde Z\,\widehat\xi_{d,\tau}\}
-
\E_d\{\widetilde Z\,\xi_{d,\tau}\}
\right\| \\
&\qquad\le
\sup_{\tau\in\T}
\left\|
\E_{n,d}\{\widetilde Z
(\widehat\xi_{d,\tau}-\xi_{d,\tau})\}
\right\|
+
\sup_{\tau\in\T}
\left\|
(\E_{n,d}-\E_d)
\{\widetilde Z\,\xi_{d,\tau}\}
\right\|.
\end{align*}
The first term is
$o_p(1)$
by \eqref{eq:app-xi-l2}, the Cauchy--Schwarz inequality, and
$\E_{n,d}\|\widetilde Z\|^2=O_p(1)$.
The second term is also
$o_p(1)$
because
$\{\widetilde Z\,\xi_{d,\tau}:\tau\in\T\}$
is a Glivenko--Cantelli class with integrable envelope
$C\|\widetilde Z\|(1+|Y|)$.
Together with
$\lambda_n\to\lambda$,
these results give
\begin{equation}\label{eq:app-gamma-consistency}
\sup_{\tau\in\T}
\|\widehat C_\xi(\tau)-C_\xi(\tau)\|
=
o_p(1),
\qquad
\sup_{\tau\in\T}
\|\widehat\gamma(\tau)-\gamma(\tau)\|
=
o_p(1).
\end{equation}
The second conclusion also uses
$\widehat\Sigma_Z\xrightarrow{p}\Sigma_Z$
and the positive definiteness of
$\Sigma_Z$.
Positive definiteness further implies that
$\widehat\Sigma_Z$
is invertible with probability approaching one.

Using
\eqref{eq:S-linear},
\eqref{eq:app-B-linearization},
and
\eqref{eq:app-gamma-consistency},
we obtain
\begin{align*}
S_n^{\mathrm{adj}}(\tau)
={}&
S_n(\tau)-\gamma(\tau)^\top B_n
-
\{\widehat\gamma(\tau)-\gamma(\tau)\}^\top B_n \\
={}&
\sqrt N\,r(\tau)
+
\sqrt{1-\lambda_n}\,
\mathbb G_{n,1}
\{\xi_{1,\tau}-\gamma(\tau)^\top\widetilde Z\} \\
&-
\sqrt{\lambda_n}\,
\mathbb G_{n,0}
\{\xi_{0,\tau}-\gamma(\tau)^\top\widetilde Z\}
+
o_p(1),
\end{align*}
uniformly over
$\tau\in\T$,
since
$B_n=O_p(1)$.
This proves
\eqref{eq:ca-linearization}.

The function classes
\[
\left\{
\xi_{d,\tau}-\gamma(\tau)^\top\widetilde Z:
\tau\in\T
\right\}
\]
are Donsker under the joint arm-$d$ distribution of
$(Y(d),Z)$.
Indeed,
$\{\xi_{d,\tau}\}$
is Donsker by Lemma~\ref{lem:appendix-regularity}, while
$\{\gamma(\tau)^\top\widetilde Z:\tau\in\T\}$
is contained in the finite-dimensional linear span of
$\widetilde Z$
with bounded coefficients.
The Donsker property and independence of the two samples therefore yield the Gaussian limit stated in the theorem.

Under $H_0$, the covariance kernel of the unadjusted Gaussian process is
$\Gamma(\tau,s)$.
Write
\[
C_\xi(\tau)
=
(1-\lambda)\E_1\{\widetilde Z\,\xi_{1,\tau}\}
+
\lambda\E_0\{\widetilde Z\,\xi_{0,\tau}\}.
\]
Using
$\gamma(\tau)=\Sigma_Z^{-1}C_\xi(\tau)$,
we obtain
\begin{align*}
\Gamma_{\mathrm{adj}}(\tau,s)
&=
\Gamma(\tau,s)
-\gamma(\tau)^\top C_\xi(s)
-C_\xi(\tau)^\top\gamma(s)
+\gamma(\tau)^\top\Sigma_Z\gamma(s) \\
&=
\Gamma(\tau,s)
-
C_\xi(\tau)^\top\Sigma_Z^{-1}C_\xi(s),
\end{align*}
which proves
\eqref{eq:ca-kernel}.
Finally, for any
$\tau_1,\ldots,\tau_J\in\T$
and
$a_1,\ldots,a_J\in\mathbb R$,
\[
\sum_{j,k=1}^J
a_ja_k
\{\Gamma-\Gamma_{\mathrm{adj}}\}(\tau_j,\tau_k)
=
\left\|
\Sigma_Z^{-1/2}
\sum_{j=1}^J a_j C_\xi(\tau_j)
\right\|^2
\ge0.
\]
Therefore,
$\Gamma_{\mathrm{adj}}\preceq\Gamma$.
\end{proof}

\clearpage
\section{Additional Monte Carlo Results}
\label{app:additional-monte-carlo}

The main text reports results under balanced assignment. This section examines moderate treatment-control imbalance using $(n_1,n_0)\in\{(50,100),(100,200),(200,400)\}$, corresponding to treatment share $n_1/n=1/3$. The qualitative size conclusions are similar to those under balanced assignment: eMBT remains close to the nominal level in the three symmetric designs, while the smaller lognormal samples are more demanding. The principal power patterns are likewise qualitatively similar to the balanced case across the scale and Normal-to-$t_5$ alternatives.

\begin{table}[htbp]
\centering
\begin{threeparttable}
\caption{Empirical size under moderately unbalanced assignment}
\label{tab:unbalanced-size}
\small
\setlength{\tabcolsep}{1.6pt}
\renewcommand{\arraystretch}{1.04}
\begin{tabular}{llrrrrrrrr}
\toprule
&&\multicolumn{4}{c}{eMBT}&\multicolumn{4}{c}{Comparison tests}\\
\cmidrule(lr){3-6}\cmidrule(lr){7-10}
Innovation & $(n_1,n_0)$ & KS & CvM & S--KS & S--CvM & FRT-CI & RS &
D-mtPT & Q-mtPT\\
\midrule
\multirow{3}{*}{Normal}
&(50,100)&0.056&0.064&0.068&0.063&0.030&0.067&0.041&0.038\\
&(100,200)&0.044&0.048&0.056&0.055&0.035&0.069&0.046&0.048\\
&(200,400)&0.038&0.044&0.044&0.049&0.038&0.082&0.047&0.042\\
\addlinespace
\multirow{3}{*}{$t_5$}
&(50,100)&0.045&0.047&0.058&0.050&0.019&0.037&0.042&0.050\\
&(100,200)&0.051&0.058&0.060&0.055&0.031&0.034&0.034&0.047\\
&(200,400)&0.039&0.041&0.043&0.047&0.033&0.023&0.036&0.026\\
\addlinespace
\multirow{3}{*}{Lognormal}
&(50,100)&0.073&0.068&0.078&0.064&0.039&0.029&0.136&0.164\\
&(100,200)&0.085&0.073&0.068&0.065&0.036&0.020&0.096&0.102\\
&(200,400)&0.055&0.054&0.045&0.047&0.030&0.030&0.102&0.086\\
\addlinespace
\multirow{3}{*}{Product normal}
&(50,100)&0.060&0.054&0.082&0.058&0.032&0.048&0.053&0.069\\
&(100,200)&0.063&0.059&0.068&0.064&0.043&0.030&0.056&0.057\\
&(200,400)&0.056&0.057&0.060&0.058&0.032&0.034&0.052&0.051\\
\bottomrule
\end{tabular}
\begin{tablenotes}[flushleft]
\footnotesize
\item Notes: Entries are rejection frequencies at the $5\%$ nominal level
based on 1,000 Monte Carlo replications of \eqref{eq:simulation-dgp} with
$\sigma_\delta=0$ and $n_1/n=1/3$.  S denotes pointwise studentization;
RS denotes recentered subsampling; D-mtPT and Q-mtPT denote the CDF- and
quantile-based martingale-transformed permutation tests, respectively.  The
latter two are implemented using \texttt{RATest} 0.1.12.
\end{tablenotes}
\end{threeparttable}
\end{table}
\FloatBarrier

\begin{table}[p]
\centering
\begin{threeparttable}
\caption{Empirical power under moderately unbalanced assignment}
\label{tab:unbalanced-power}
\scriptsize
\setlength{\tabcolsep}{1.6pt}
\renewcommand{\arraystretch}{1.00}
\begin{tabular}{llrrrrrrrr}
\toprule
&&\multicolumn{4}{c}{eMBT}&\multicolumn{4}{c}{Comparison tests}\\
\cmidrule(lr){3-6}\cmidrule(lr){7-10}
Innovation & $(n_1,n_0)$ & KS & CvM & S--KS & S--CvM & FRT-CI & RS &
D-mtPT & Q-mtPT\\
\midrule
\multicolumn{10}{l}{\textit{Panel A: $\sigma_\delta=0.2$}}\\[2pt]
\multirow{3}{*}{Normal}
&(50,100)&0.217&0.264&0.242&0.259&0.072&0.168&0.122&0.172\\
&(100,200)&0.457&0.510&0.471&0.506&0.151&0.389&0.221&0.258\\
&(200,400)&0.754&0.806&0.761&0.785&0.342&0.633&0.451&0.460\\
\addlinespace
\multirow{3}{*}{$t_5$}
&(50,100)&0.125&0.157&0.171&0.172&0.058&0.097&0.117&0.137\\
&(100,200)&0.226&0.310&0.294&0.333&0.115&0.141&0.180&0.233\\
&(200,400)&0.424&0.525&0.525&0.579&0.215&0.275&0.273&0.432\\
\addlinespace
\multirow{3}{*}{Lognormal}
&(50,100)&0.079&0.077&0.078&0.068&0.087&0.048&0.270&0.246\\
&(100,200)&0.080&0.076&0.103&0.087&0.127&0.040&0.325&0.271\\
&(200,400)&0.093&0.122&0.171&0.150&0.204&0.088&0.402&0.356\\
\addlinespace
\multirow{3}{*}{Product normal}
&(50,100)&0.085&0.095&0.114&0.108&0.040&0.079&0.065&0.129\\
&(100,200)&0.151&0.197&0.193&0.218&0.046&0.097&0.056&0.154\\
&(200,400)&0.259&0.345&0.311&0.355&0.054&0.177&0.060&0.222\\
\midrule
\multicolumn{10}{l}{\textit{Panel B: $\sigma_\delta=0.5$}}\\[2pt]
\multirow{3}{*}{Normal}
&(50,100)&0.773&0.831&0.809&0.824&0.350&0.619&0.299&0.519\\
&(100,200)&0.982&0.989&0.985&0.988&0.752&0.923&0.643&0.842\\
&(200,400)&1.000&1.000&1.000&1.000&0.978&0.997&0.955&0.983\\
\addlinespace
\multirow{3}{*}{$t_5$}
&(50,100)&0.469&0.637&0.653&0.671&0.210&0.369&0.280&0.454\\
&(100,200)&0.788&0.909&0.909&0.927&0.498&0.579&0.517&0.747\\
&(200,400)&0.984&0.999&0.999&0.999&0.888&0.906&0.854&0.975\\
\addlinespace
\multirow{3}{*}{Lognormal}
&(50,100)&0.134&0.143&0.191&0.173&0.220&0.110&0.485&0.547\\
&(100,200)&0.204&0.290&0.425&0.372&0.470&0.201&0.698&0.687\\
&(200,400)&0.358&0.492&0.750&0.655&0.790&0.543&0.918&0.887\\
\addlinespace
\multirow{3}{*}{Product normal}
&(50,100)&0.275&0.348&0.387&0.390&0.054&0.254&0.109&0.271\\
&(100,200)&0.532&0.673&0.660&0.711&0.049&0.316&0.097&0.476\\
&(200,400)&0.861&0.936&0.936&0.955&0.073&0.604&0.166&0.723\\
\midrule
\multicolumn{10}{l}{\textit{Panel C: variance-matched normal-to-$t_5$ alternative}}\\[2pt]
\multirow{3}{*}{Normal/$t_5$}
&(50,100)&0.098&0.107&0.129&0.119&0.061&0.094&0.037&0.113\\
&(100,200)&0.096&0.124&0.150&0.146&0.129&0.160&0.041&0.215\\
&(200,400)&0.108&0.173&0.216&0.231&0.286&0.285&0.087&0.400\\
\midrule
\multicolumn{10}{l}{\textit{Panel D: unstandardized normal-to-$t_5$ alternative}}\\[2pt]
\multirow{3}{*}{Normal/$t_5$}
&(50,100)&0.169&0.182&0.189&0.184&0.038&0.125&0.067&0.091\\
&(100,200)&0.352&0.386&0.342&0.369&0.058&0.202&0.098&0.107\\
&(200,400)&0.654&0.706&0.649&0.656&0.107&0.367&0.203&0.168\\
\bottomrule
\end{tabular}
\begin{tablenotes}[flushleft]
\footnotesize
\item Notes: Entries are rejection frequencies at the $5\%$ nominal level
based on 1,000 Monte Carlo replications of \eqref{eq:simulation-dgp} with
$n_1/n=1/3$ in Panels A and B and of
\eqref{eq:simulation-shape-dgp} in Panels C and D.  They are not adjusted for
differences in empirical size.  S denotes pointwise studentization; RS denotes
recentered subsampling; D-mtPT and Q-mtPT denote the CDF- and quantile-based
martingale-transformed permutation tests, respectively.  The latter two are
implemented using \texttt{RATest} 0.1.12.
\end{tablenotes}
\end{threeparttable}
\end{table}
\FloatBarrier

\FloatBarrier
\clearpage
\putbib
\end{bibunit}

\end{document}